\documentclass[11pt]{article}
\usepackage[letterpaper]{geometry}
\usepackage{ltexpprt}

\usepackage[utf8]{inputenc} % allow utf-8 input
\usepackage[T1]{fontenc}    % use 8-bit T1 fonts
\usepackage{hyperref}       % hyperlinks
\usepackage{url}            % simple URL typesetting
\usepackage{booktabs}       % professional-quality tables
\usepackage{amsfonts}       % blackboard math symbols
\usepackage{nicefrac}       % compact symbols for 1/2, etc.
\usepackage{microtype}      % microtypography
\usepackage{amsmath,amssymb}
\usepackage{algorithm,algpseudocode}
\usepackage{booktabs}
\usepackage{comment}
\usepackage{wrapfig}
\usepackage{enumerate}
\usepackage[skip=2pt]{caption}

\usepackage{pgfplots}
\pgfplotsset{compat=newest}
\usepackage{tikz}
\usetikzlibrary{patterns}
\usetikzlibrary{calc}

\newtheorem{remark}{Remark}
\newtheorem{problem}{Problem}

\numberwithin{equation}{section}
\numberwithin{figure}{section}
\numberwithin{algorithm}{section}
\newcommand{\floor}[1]{\lfloor #1 \rfloor}
\newcommand{\argmax}{\operatornamewithlimits{argmax}}

\newcommand{\E}{\mathbb{E}}
\newcommand{\lat}{\mathcal{L}}
\newcommand{\R}{\mathbb{R}}

\newcommand{\OPT}{\mathrm{OPT}}
\newcommand{\smallset}{\mathrm{small}}
\newcommand{\Prob}{\mathrm{P}}
\newcommand{\V}{\mathbb{V}}
\newcommand{\adm}[1]{\mathrm{adm}({#1})}
\newcommand*{\diffdchar}{\mathrm{d}}   
\newcommand*{\dd}{\mathop{\diffdchar\!}}
\newcommand{\curve}{u}
\newcommand{\ulm}[2]{\curve_{{#1},{#2}}}
\newcommand{\Cube}{\mathrm{Cube}}

\usepackage{color}
\usepackage{CJKutf8}
\usepackage{ifthen}
\newcommand{\COMM}[2]{{
\begin{CJK}{UTF8}{ipxm}
\ifthenelse{\equal{#1}{SN}}{\color{blue}}{
\ifthenelse{\equal{#1}{TM}}{\color{red}}{
\ifthenelse{\equal{#1}{YY}}{\color{magenta}}{
\ifthenelse{\equal{#1}{BB}}{\color{cyan}}}}}
[#1: #2]
\end{CJK}
}}

\begin{document}

%\begin{frontmatter}

\title{Multiple Knapsack-Constrained Monotone DR-Submodular Maximization on Distributive Lattice \\ --- Continuous Greedy Algorithm on Median Complex ---}
% alphabetical order
%\author[1]{Takanori Maehara}
%\address[1]{RIKEN AIP}
%\author[2]{So Nakashima}
%\address[2]{The University of Tokyo}
%\author[3]{Yutaro Yamaguchi}
%\address[3]{Osaka University}
% 足させてください
%\date{\today}
\author{
Takanori Maehara\thanks{RIKEN AIP.}
\and 
So Nakashima\thanks{The University of Tokyo.}
\and
Yutaro Yamaguchi$^*$\thanks{Osaka University.}}
\date{}
\maketitle

\begin{abstract}
We consider a problem of maximizing a monotone DR-submodular function under multiple order-consistent knapsack constraints on a distributive lattice. 
Since a distributive lattice is used to represent a dependency constraint, the problem can represent a dependency constrained version of a submodular maximization problem on a set.
We propose a $1 - 1/e$ approximation algorithm for this problem.
To achieve this result, we generalize the continuous greedy algorithm to distributive lattices:
We choose a \emph{median complex} as a continuous relaxation of a distributive lattice and define the \emph{multilinear extension} on it.
We show that the median complex admits special curves, named \emph{uniform linear motions}, such that the multilinear extension of a DR-submodular function is concave along a positive uniform linear motion, which is a key property of the continuous greedy algorithm.
\end{abstract}

%\end{frontmatter}

\tableofcontents
\thispagestyle{empty}
\setcounter{page}{0}

\clearpage 

\section{Introduction}
\label{sec:introduction}

\subsection{Problem and Result}

In this study, we consider multiple knapsack-constrained monotone DR-submodular maximization problem on finite distributive lattices.

Let $(\lat, \le)$ be a finite distributive lattice, and let $J(\lat)$ be the set of join-irreducible elements of $\lat$. 
For $X \in \lat$, we denote by $\adm{X}$ the set of join-irreducible elements that is admissible to $X$ (see Section~\ref{sec:preliminaries} for definitions).
A function $f \colon \lat \to \mathbb{R}$ is \emph{monotone} if $f(X) \le f(Y)$ for all $X \le Y$. 
$f$ is a \emph{DR-submodular function}~\cite{gottschalk2015submodular,nakashima2019subspace} if
\begin{align}
    f(X \lor p) - f(X) \ge f(Y \lor q) - f(Y)
\end{align}
for all $X, Y \in \lat$ and $p, q \in J(\lat)$ such that $X \le Y$, $p \le q$, $p \in \adm{X}$, and $q \in \adm{Y}$.
Let $c \colon J(\lat) \to \mathbb{R}_+$ be the weight on the join-irreducible elements.
We identify $c$ as the function $c \colon \lat \to \mathbb{R}$ by $c(X) = \sum_{p \le X} c(p)$ the weight of $X$. 
We say that $c$ is \emph{order-consistent} if $c(p) \le c(q)$ for all $p, q \in J(\lat)$ with $p \le q$.
A \emph{knapsack constraint} is represented by $c(X) \le b$ for $b \in \mathbb{R}_+$.
Then, the \emph{multiple knapsack-constrained monotone DR-submodular maximization problem} is the following optimization problem:
\begin{align}
\label{eq:problem}
    \begin{array}{ll}
         \text{maximize} & f(X)  \\
         \text{subject to} & c_\lambda(X) \le b_\lambda, \ (\lambda \in \Lambda),
    \end{array}
\end{align}
where $f \colon \lat \to \mathbb{R}$ is a monotone DR-submodular function, $\Lambda$ is a finite set with $|\Lambda| = O(1)$, and $c_\lambda(X) \le b_\lambda$ is a knapsack constraint with order-consistent $c_\lambda$ for each $\lambda \in \Lambda$.

We prove that this problem is solvable in polynomial time within an approximation factor of $1 - 1/e$.
This is our main theorem.
\begin{theorem}
\label{thm:main}
For any $\epsilon > 0$, there exists a polynomial-time $(1 - \epsilon) (1 - 1/e)$-approximation algorithm for the multiple knapsack-constrained monotone DR-submodular maximization problem.
\end{theorem}

This problem naturally arises in ``dependency-constrained'' problems as follows.
\begin{Example}
Consider a sensor activation problem: Let $P$ be a set of sensors placed on a space, and we want to activate a subset of sensors that maximizes the coverage area under some constraint.
This is a typical problem of maximizing a monotone submodular set function.
Now we consider dependency constraints represented by a directed acyclic graph $\mathcal{G} = (P, E)$, where each edge $(p, q) \in E$ means that sensor $p$ can be selected only if sensor $q$ is already selected. 
Then, the set of possible sensor selection $X \subseteq P$ forms a distributive lattice.
If all the marginal covered area of $p$ is smaller than that of $q$ for all $(p, q) \in E$, the function $f$ forms a monotone DR-submodular function on this lattice.
Also if the cost $c(p)$ is more expensive $c(q)$ if $(p, q) \in E$, the cost function $c$ is order-consistent.
Thus, if these conditions are met, the problem can be represented as a (multiple) knapsack-constrained monotone DR-submodular maximization problem.
\end{Example}

\subsection{Background and Motivation}

Let $V$ be a finite set. 
A function $f \colon 2^V \to \mathbb{R}$ is \emph{submodular} if it satisfies the submodular inequality: for all $X, Y \subseteq V$, 
\begin{align}
\label{eq:set-submodular}
    f(X) + f(Y) \ge f(X \cap Y) + f(X \cup Y).
\end{align}
$f$ is submodular if and only if it satisfies the \emph{diminishing return property}: for all $X, Y \subseteq V$ with $X \subseteq Y$ and $p \in V \setminus Y$,
\begin{align}
\label{eq:set-diminishing-return}
    f(X \cup p) - f(X) \ge f(Y \cup p) - f(Y).
\end{align}
Submodular functions are ubiquitous in many fields, including combinatorial optimization~\cite{fujishige1983note}, economics~\cite{murota2016discrete}, and machine learning~\cite{krause2014submodular}.
Thus, maximizing such a function is regarded as one of the fundamental combinatorial optimization problems.
The problem is NP-hard in general~\cite{cornnejols1977location}; however,  after a seminal work by Nemhauser, Wolsey, and Fisher~\cite{nemhauser1978analysis}, many approximation algorithms have been proposed for several types of constraints such as cardinality constraint~\cite{nemhauser1978analysis}, matroid constraint~\cite{calinescu2011maximizing}, knapsack constraint~\cite{sviridenko2004note}, and multiple knapsack constraint~\cite{kulik2009maximizing}.

Several attempts have been conducted to generalize the domain of the problem from a finite set to more general space, say, a lattice.
In general, submodular functions can be defined via the lattice submodularity: Let $\lat$ be a lattice. Then, $f \colon \lat \to \mathbb{R}$ is lattice submodular~\cite{topkis1978minimizing} if for all $X, Y \in \lat$,
\begin{align}
\label{eq:lattice-submodular}
    f(X) + f(Y) \ge f(X \land Y) + f(X \lor Y).
\end{align}

For the integer lattice $\mathbb{Z}^n$, 
Alon et al.~\cite{alon2012optimizing} studied a particular monotone submodular function on the integer lattice and proposed $1 - 1/e$ approximation algorithm to maximize the function under a knapsack constraint.
Soma and Yoshida~\cite{soma2018maximizing} generalized this technique for general submodular functions on the integer lattice.
They also introduce \emph{DR-submodular functions} as follows: for $X, Y \in \mathbb{Z}^n$ and $j \in \{1, \dots, n\}$, 
\begin{align}
\label{eq:int-diminishing-return}
    f(X + e_j) - f(X) \ge f(Y + e_j) - f(Y),
\end{align}
where $e_j$ is the $j$-th unit vector.
For a set function, the submodularity~\eqref{eq:set-submodular} and the diminishing return property~\eqref{eq:set-diminishing-return} is equivalent; however, it is not the case on integer lattices, i.e., \eqref{eq:lattice-submodular} does not imply \eqref{eq:int-diminishing-return}.
The DR-submodularity is obtained by generalizing of the diminishing return property.
The relation between the DR-submodular functions and submodular set functions are explored in \cite{ene2016reduction}.

For more general lattices, 
Gottschalk and Peis~\cite{gottschalk2015submodular} generalized the DR-submodularity on distributive lattices.
They proved that the cardinality-constrained problem can be solved within an approximation factor of $1 - 1/e$.
Also, they showed that the knapsack-constrained problem is hard to approximate in general.
The first and second authors~\cite{nakashima2019subspace} generalized the DR-submodularity to modular lattices and prove that cardinality constrained problem and order-consistent knapsack-constrained problem can be solved within constant approximation factors.

Thus far, all the existing submodular maximization algorithms on general lattices are ``combinatorial,'' i.e., the greedy algorithm.
On the other hand, the most powerful algorithm for submodular maximization on sets is based on the continuous relaxation~\cite{chekuri2014submodular}.
For a submodular set function $f \colon 2^V \to \mathbb{R}$, its \emph{multilinear extension}~\cite{vondrak2008optimal,calinescu2011maximizing} $F \colon [0, 1]^V \to \mathbb{R}$ is defined by
\begin{align}
    F(x) = \E_{\hat X \sim x}[ f(\hat X) ] = \sum_{X \subseteq V} f(X) \prod_{i \in X} x_i \prod_{j \in V \setminus X} (1 - x_j),
\end{align}
where $\E_{\hat X \sim x}$ is the expectation with respect to the probability distribution such that $\Prob[i \in \hat X] = x_i$ for all $i \in V$ independently.
This function is not a concave function but can be maximized within a constant approximation factor under some conditions by the continuous greedy algorithm~\cite{vondrak2008optimal,calinescu2011maximizing} or projected gradient method~\cite{hassani2017gradient}.
Once we obtain a continuous approximate solution, we can obtain a discrete approximate solution via rounding method such as the pipage rounding~\cite{ageev2004pipage} or the contention resolution scheme~\cite{chekuri2014submodular}.
Soma and Yoshida~\cite{soma2018maximizing} generalized the multilinear extension to a DR submodular function on the integer lattice.
However, this approach has not been generalized to general lattices.
Thus, we had the following research question:
\begin{problem}
Can we generalize the continuous greedy algorithm to general lattices?
\end{problem}

This study gives the first positive answer to this question.
Our main result (Theorem~\ref{thm:main}) is obtained by the continuous greedy algorithm on distributive lattices.

\subsection{Proof Outline}

Let us recall Kulik, Shachnai, and Tamir~\cite{kulik2009maximizing}'s algorithm for the multiple knapsack-constrained submodular maximization problem on sets.
Their algorithm uses the \emph{continuous greedy algorithm} to find a good approximate continuous solution.
Then, it uses \emph{partial enumeration} to obtain a good approximate discrete solution.
%\COMM{SN}{partial enumeration は rounding の方に効いてる印象なので，
%"combines conti. greedy alg. and partial enumeration. Conti. greedy alg. finds a good apprximate...  Then, they employed partial enumeration to obtain a good" とかですかね？
%%"to find a good approximate continuous solution and to rounds it " 
%}
%\COMM{SN}{
%$\uparrow$ 読み間違えてたのかもしれないのですが，「to find a good approximate continuous solution」が continuous greedy だけにかかっているなら異存はないです．
%
%（Sec.5のアルゴリズムにおいて "find a good approximate continuous solution" は Thm. \ref{thm:conti-greedy-for-knapsack} だけから保証されているので，partial enumeration まで修飾してしまうとまずいという認識でした．）
%}
%\COMM{TM}{厳密には round しやすい continuous solution を partial enumeration で作るので分けにくいですね．ざっくり対応．}
Our proof generalizes their proof to distributive lattices.
The main difficulty is generalizing the continuous relaxation to distributive lattices.

Let $\lat$ be a distributive lattice.
By the Birkhoff representation theorem~\cite{gratzer2002general}, the ideals (i.e., the downward-closed sets) of the poset $P$ of the join-irreducible elements of $\lat$ is isomorphic to $\lat$.
This motivates us to define the continuous domain, $K(\lat)$, as a subset of $[0, 1]^P$. 
%Using this property, we first define a continuous relaxation of a distributive lattice denoted by $K(\lat)$.
The obtained domain, $K(\lat)$, forms a \emph{cubical complex} (more precisely, \emph{median complex}), which is locally isomorphic to a Boolean hypercube $[0, 1]^n$, and is obtained by gluing the hypercubes by these faces. 
We define the multilinear extension $F \colon K(\lat) \to \mathbb{R}$ of a lattice DR-submodular function $f$ by gluing the multilinear extensions of $f$ in each hypercube.

In the set submodular maximization, the most important property of the continuous relaxation is the \emph{concavity along any positive direction}, i.e., for any two points $x, y \in [0, 1]^V$ with $x \le y$, the function $F((1 - t) x + t y)$ is concave in $t$. 
We generalize this property as follows. 
As a generalization of the straight line $[0,1] \ni t \mapsto (1 - t) x + t y$, we introduce \emph{uniform linear motion} $c_{x,y} \colon [0,1] \to K(\lat)$ from $x \in K(\lat)$ to $y \in K(\lat)$.
The uniform linear motion coincides with a straight line in each hypercube.
At a point that belongs to the common face of two hypercubes, the velocity of the uniform linear motion must satisfy the ``flow conservation law.''
We prove that for any two points $x, y \in K(\lat)$, there uniquely exists uniform linear motion from $x$ to $y$ (Theorem~\ref{thm:unique-existence-uniform-linear-motion}).
Also, we prove that the function $F(c_{x,y}(t))$ is concave in $t$ (Theorem~\ref{thm:multilinear-concave}).
Using this property, we can generalize the continuous greedy algorithm (Theorem~\ref{thm:conti-greedy-for-knapsack}) that has an approximation factor of $1 - 1/e - \epsilon$.

The rounding part is a generalization of Kulik, Shachnai, and Tamir~\cite{kulik2009maximizing}'s algorithm.
However, there are several minor difficulties caused by a distributive lattice.
The proof in Section~\ref{sec:rounding} verifies the generalization is valid.

\subsection{Other Related Work}

Submodular ``minimization,'' instead of the maximization, is also a well-studied problem~\cite{fujishige1983note}. 
The key technique in the submodular minimization is the \emph{Lov\'asz extension}, which is another continuous relaxzation defined on $[0,1]^V$.
A set function is submodular if and only if its Lov\'asz extension is convex~\cite{lovasz1983submodular}.
%\COMM{SN}{"Therefore, we can apply convex optimization technique to submodular minimization" まで書くか？（くどい？）}

Submodular minimization has been generalized to lattices~\cite{topkis1978minimizing}.
%Hirai~\cite{hirai2018convexity} introduced a continuous object, called \emph{orthoscheme complex}, as a continuous relaxation of a modular lattice.
%\COMM{SN}{
Brady and McCammond~\cite{brady2010braids} introduced \emph{orthoscheme complex} in the context of geometric group theorey. %, which is useful for the continuous relaxation of submodular minimization over general lattices.
%Hirai~\cite{hirai2018convexity} showed that 
%}
Chalopin, Chepoi, Hirai, and Osajda~\cite{chalopin2014weakly} showed that orthoscheme complex of modular lattices with $L_2$ norm forms a \emph{CAT(0)-space}~\cite{bridson1999metric}, which admits unique geodesic for any two points.
This also makes the orthoscheme complex as a \emph{geodesic convex space}.
Hirai~\cite{hirai2018convexity} showed that a function on a modular lattice is submodular if and only if its Lov\'asz extension (generalized to a lattice) is geodesic convex.

Our approach differs from this line.
To define the multilinear extension of submdoular function over distributive lattices, we use median complex, which is different from orthscheme complex.

Chepoi~\cite{chepoi2000graphs} and Roller~\cite{roller1998poc} independently proved that a median complex equipped with $l_2$-metric also forms a CAT(0)-space.
Our uniform linear motion is different from the geodesic in this context (see Example in Section~\ref{subsec:ulm}).

%The most important property for submodular maximization is the concavity along any positive direction. 
%To obtain this property, we have to introduce a geometric structure (i.e., modular complex with uniform linear motions) that is different from the CAT(0) space.
%A median complex equipped with $L_2$ distance forms a CAT(0) space~\cite{chepoi2000graphs, roller1998poc}; however, the uniform linear motion is not a geodesic; hence we cannot use the theory of CAT(0) space.

%\COMM{SN}{
%$\uparrow$ パラグラフについて：metric以前にそもそも空間として違うので，
%To define multilinear extension of submdoular function over distributive lattice, we use median complex, which is different from orthscheme complex.
%
%後半部のmedian complex + $L_2$-metric とは違う構造（median complex + ULM）を使っていることも別途どこかで強調したい．（submodular minimization の話ではないので別のどこかに移したい．)Sec 1.3の第3パラグラフの最後に
%(We remark that another convex structure over median complex is known... (median complex + $L_2$の説明．別の直線構造が入ると説明．)... However, we do not use this structure to make $F(c_{x,y}(t))$ concave as explained earlier.
%}
\section{Preliminaries}
\label{sec:preliminaries}

%\COMM{SN}{
%notationについて．
%
%\begin{itemize}
%    \item $p \in P$ （$i \in P$から変更．$i$を自然数の添え字として使いたいのと$p$の方が文字がそろうので）．
%    \item $P$の要素は小文字（$p \in P$）．部分集合は大文字（$X \subseteq P$）．部分集合の確率変数は$\hat X$とする．（ふつうは小文字--大文字で確率変数を区別するが，今回は大文字を部分集合に充てているので hat をつける．）
%    \item $c,C$が色々なところで出てきていたので適宜変更．
%    \item $\Prob[\text{<event>}]$
%\end{itemize}
%}

\subsection{Distributive Lattice}

Let $(P, \le)$ be a poset.We say $p \in P$ \emph{covers} $p' \in P$ if $p' < p$ and there are no $p'' \in P$ such that $p' \le p'' \le p$.
We denote this relation by $p' \prec p$.
A subset $I \subseteq P$ is an \emph{ideal} if $I$ is downward closed: $p' \le p$ and $p \in I$ implies that $p' \in I$.
For $p \in P$, let $I_p$ be the principal ideal $\{p' \in P \mid p' \le p \}$.
A subset $X \subseteq P$ is an \emph{antichain} if all pairs in $X$ are incomparable.

A \emph{lattice} $(\lat, \le)$ is a partially ordered set with the largest common lower bound $X \land Y$ and the least common upper bound $X \land Y$ for any $X,Y \in \lat$.
The former is called \emph{meet} and the latter $\emph{join}$.
An element $X \in \lat$ is \emph{join-irreducible} if there is the unique $Y \prec X$.
A lattice $\lat$ is said to be \emph{distributive} if $X \land (Y \lor Z) = (X \lor Y) \land (X \lor Z)$ and $X \lor (Y \land Z) = (X \lor Y) \land (X \lor Z)$ for any $X,Y,Z \in \lat$.
In this paper, we only deal with finite distributive lattices.

There is a one-to-one correspondence between a distributive lattice $\lat$ and a poset $P$ up to isomorphism owing to the Birkhoff representation theorem.
Indeed, the family $I(P)$ of the ideals of $P$ is a distributive lattice with meet $X \land Y := X \cap Y$ and join $X \lor Y := X \cup Y$.
Conversely, We can construct a poset $J(\lat)$ of the join-irreducible elements of $\lat$, where the order is inherited from $\lat$.
Then, Birkhoff representation theroem states that $P \cong J(I(P))$ and $\lat \cong I(J(\lat))$.
In the rest of the paper, $\lat$ denotes a distributive lattice and $P$ the corresponding poset $J(\lat)$.

\subsection{DR-Submodular Function}

Let $\lat$ be a distributive lattice.
Let $\adm{X}$ denote the set of the minimal elements of $P \setminus X$, which is called \emph{admissible elements}.
A function $f \colon \lat \to \mathbb{R}$ is a \emph{DR-submodular function} if
\begin{align}
    f(X \lor a) - f(X) \le f(Y \lor b) - f(Y)
\end{align}
for all $X, Y \in \lat$ and $a, b \in J(\lat)$ such that $X \le Y$, $a \le b$, and $a \in \mathrm{adm}(X)$, $b \in \mathrm{adm}(Y)$.
A function $f \colon \lat \to \mathbb{R}$ is a \emph{DR-supermodular function} if $-f$ is a DR-submodular function.

\subsection{Notation}
For a vector $x \in \mathbb{R}^n$, let $x_p$ be the $p$-th component of $x$.
We define $\mathrm{sign} \colon \mathbb{R} \to \{-1,1\}$ by $\mathrm{sign}(w) = 1$ if $w \ge 0$ and $\mathrm{sign}(w) = -1$ otherwise.
Let $h$ be a function of $\epsilon$.
We write one-sided limit to $0$ from above as
$\lim_{\epsilon \to 0^+} h(\epsilon)$.
For a function $f \colon \mathbb{R} \to \mathbb{R}^n$, let
\begin{align}
    &\frac{\dd^+ f(t)}{\dd t} = \lim_{\epsilon \to 0^+} \frac{f(t + \epsilon) - f(t)}{\epsilon},\\
    &\frac{\dd^- f(t)}{\dd t} = \lim_{\epsilon \to 0^+} \frac{f(t) - f(t-\epsilon)}{\epsilon}.
\end{align}
For a multivariate function $f \colon\mathbb{R}^m \to \mathbb{R}^n$ and one of the coordinate $p$, let
\begin{align}
     &\frac{\partial^+ f(x)}{\partial x_p} = \lim_{\epsilon \to 0^+} \frac{f(x + \epsilon e_p) - f(x)}{\epsilon},\\
     &\frac{\partial^- f(x)}{\partial x_p} = \lim_{\epsilon \to 0^+} \frac{f(x) - f(x - \epsilon e_p)}{\epsilon},
\end{align}
where $e_p$ is the unit vector for the $p$-th coordinate.

For a function $f \colon \lat \to \mathbb{R}$ and ideal $T\subseteq P$, we denote $f(I(T))$ as $f(T)$ owing to the Birkhoff representation theorem.
For ideals $S,T$, let $f_T(S) = f(T \cup S) - f(T)$.

\section{Continuous Extension of Distributive Lattices}

\subsection{Median Complex Arose From a Distributive Lattice}
In this section, we introduce a continuous extension of distributive lattices.
A \emph{continuous extension} $K(\lat)$ of a distributive lattice $\lat = I(P)$ is defined as
\begin{align}
    K(\lat) = \{x \in [0,1]^P \mid (x_p > 0 \land p' < p) \Rightarrow x_{p'} = 1\}.
\end{align}

We will see the properties of $K(\lat)$.
\begin{lemma}
\label{lem:support-forms-an-ideal}
A point $x \in [0,1]^P$ is in $K(\lat)$ if and only if its support $\mathrm{supp}(x) = \{ p \in P : x_p > 0 \}$ forms an ideal and $x_p = 1$ for non-maximal element $p$ of $\mathrm{supp}(x)$.
\end{lemma}
\begin{proof}
This immediately follows from the definition of $K(\lat)$.
\end{proof}

\begin{lemma}
\label{lem:continuous-extension-is-distributive-lattice}
The continuous extension $K(\lat)$ forms a distributive lattice with respect to the element-wise inequality, i.e., $x \le y$ if and only if $x_p \le y_p$ for all $p \in P$.
\end{lemma}
\begin{proof}
We show that $K(\lat)$ forms a distributive lattice equipped with the following meet and join operators:
\begin{align}
    x \land y &= (\min\{x_1, y_1\}, \min\{x_2, y_2\}, \dots, \min\{x_{|P|}, y_{|P|}\}),\\
    x \lor y &= (\max\{x_1, y_1\}, \max\{x_2, y_2\}, \dots, \max\{x_{|P|}, y_{|P|}\}).
\end{align}
We can easily see that the above operations are closed in $K(\lat)$ by Lemma~\ref{lem:support-forms-an-ideal} and satisfies distributive law.
\end{proof}

A set of hypercubes is a \emph{cubical complex} if (1) if $C \in \mathcal{K}$ then any face of $C$ is also in $\mathcal{C}$, (2) If $C_1, C_2 \in \mathcal{K}$ then $C_1 \cap C_2$ is a face of $C_1$ and $C_2$.

\begin{lemma}
$K(\lat)$ is a cubical complex.
\end{lemma}
The proof is essentially given by \cite{chepoi2000graphs, roller1998poc}.
For the sake of completeness, we give a proof here.
We remark that this cubical complex is called a \emph{median complex}~\cite{bandelt2008metric}.
\begin{proof}
We show this lemma by constructing a cubical complex $K'(\lat)$ congruent to $K(\lat)$.
Let $K'(\lat)$ be the cubical complex constructed as follows.
For any maximal antichain $X \subseteq P$, there is a distinct corresponding hypercube $\Cube(X) = [0,1]^X$ in $K'(\lat)$.
We also add the faces of $\Cube(X)$ to $K'(\lat)$.
We identify the points on two hypercubes by the following rules.
Two antichains $X$ and $Y$ are \emph{adjacent} if the following conditions are satisfied:
(1) for any $p \in X \setminus Y$, there exists $p' \in Y \setminus X$ such that $p \prec p'$ or $p' \prec p$.
(1') for any $p \in Y \setminus X$, there exists $p' \in X \setminus Y$ such that $p \prec p'$ or $p' \prec p$.
(2) for any $p \in  X \setminus Y$ and $p' \in Y \setminus X$ with $p \prec p'$, we have $p'' \in X$ for any $p'' \prec p'$.
(2') for any $p \in  Y \setminus X$ and $p' \in X \setminus Y$ with $p \prec p'$, we have $p'' \in Y$ for any $p'' \prec p'$.
We identify $x \in \Cube(X)$ and $y \in  \Cube(Y)$ for adjacent $X$ and $Y$ if 
(1) $x_p = y_p$ for any $p \in X \land Y$,
(2) $x_p = 1$ and $y_{p'} = 0$ for any $p \prec p'$, $p \in X$, and $p' \in Y$,
and (3) $x_p = 0$ and $y_{p'} = 1$ for any $p \succ p'$, $p \in X$, and $p' \in Y$.
We can easily see that $K'(\lat)$ is a cubical complex. 

We prove that $K(\lat)$ is congruent to $K'(\lat)$.
%Notice that there is a natural inclusion $\Cube(X) \subseteq K'(\lat)$ for any antichain $X \subseteq P$ induced by $[0,1]^X \subseteq [0,1]^P$:
We construct an  bijection $i: K'(\lat) \rightarrow K(\lat)$ by gluing inclusive maps $i_x  \colon \Cube(X) \to K(\lat)$ of all maximal antichains $X$.
we define inclusive map $i_X \colon \Cube(X) \to K(\lat)$ by
(1) $i_X(x_p) = x_p$ if $p \in X$;
(2) $i_X(x_p) = 1$ if there exists $p' \in X$ such that $p < p'$;
and (3) $i_X(x_p) = 0$ otherwise. 
This map is well-defined since $X$ is an antichain.
%This inclusive map conserves the connection of two hypercubes in $K'(P)$.
If $x \in \Cube(X) \cap \Cube(Y)$, then $i_X(x) = i_Y(x)$.
Therefore, we can glue $i_X$ for all maximal antichains $X \subseteq P$ and obtain a global map $i: K'(\lat) \rightarrow K(\lat)$.
We can easily see that this map is a bijection.
\end{proof}
In the rest of the paper, we use the notation $\Cube(X)$ in the above proof.

\subsection{Uniform Linear Motion}
\label{subsec:ulm}

On the submodular set function maximization, the following property plays a crucial role: for any $x, y \in [0,1]^V$ with $x \le y$, $F((1-t) x + t y)$ is concave in $t \in [0, 1]$.
To extend this property to distributive lattice, we need to generalize line connecting two points for $K(\lat)$.
Here, we introduce such a concept, uniform linear motion.

The fundamental difference between $K(\lat)$ and Euclidean space is the existence of the face between two hypercubes.
We need to define ``straightness'' at such faces.
Consider a curve $\curve \colon [0, 1] \to K(\lat)$ that passes a face at time $t$.
The velocity vectors $v^-$ and $v^+$ immediately before and after $x$ are defined by
\begin{align}
    \label{eq:velocity-at-face}
    v^- &= \frac{\dd^+ \curve(t)}{\dd t}, \\ 
    v^+ &= \frac{\dd^- \curve(t)}{\dd t}.
\end{align}
We say that $\curve$ is \emph{straight at $\curve(t)$} if there is a flow on a network $\mathcal{N}$ defined below.
The nodes of $\mathcal{N}$ is the disjoint union of the following two sets $P_+$ and $P_-$.
Let $P_+$ be the set of $i \in P$ with $(v_+)_i > 0$ or $(v_-)_i > 0$ and $P_-$ be the set of $i \in P$ with $(v_+)_i < 0$ or $(v_-)_i < 0$.
The edges of $\mathcal{N}$ consists of $p \rightarrow p'$ for $p \le p'$ with $p,p' \in P_+$ or $p \leftarrow p'$ for $p \le p'$ with $p,p' \in P_-$.
The capacities of the edges are $[0, \infty)$.
Nodes in the network $\mathcal{N}$ become sources or sinks by the following rules:
Node $p \in P_+$ with $(v^+)_p > 0$ is a sink  the amount of whose incoming flow is $v^+_p$;
Node $p \in P_-$ with $(v^+)_p < 0$ is a source  the amount of whose outgoing flow is $-v^+_p$;
Node $p \in P_+$ with $(v^-)_p > 0$ is a source  the amount of whose outgoing flow is $v^-_p$;
Node $p \in P_-$ with $(v^-)_p < 0$ is a sink the amount of whose incoming flow is $-v^-_p$.

A \emph{uniform linear motion} from $x \in K(\lat)$ to $y \in K(\lat)$ is a curve $\curve \colon [0,1] \to K(\lat)$ such that $\curve(0) = x$, $\curve(1) = y$,the curve $\curve$ is a line segment in each hypercube, and $\curve$ is straight at any point on a face.
The existence of a uniform linear motion between two points is non-trivial.

\begin{Example}
\label{ex:ulm}
Consider a poset $P = \{p_1, p_2, p_3, p_4\}$ with $p_2 \prec p_3$ and $p_2 \prec p_4$.
Figure~\ref{fig:ulm} shows the Hasse diagram of $P$ and the corresponding median complex $K(\lat)$.
%We denote $x_{p_i}$ by $x_i$ to simplify the notation.
The uniform linear motion $\ulm{\bot}{\top}(t)$ over $K(\lat)$ path through two maximal cubes $\Cube(\{p_1,p_2\})$ and $\Cube(\{p_1, p_3,p_4\})$.
The uniform linear motion intersects with face $\{x \in K(\lat)) \mid  x_{p_2} = 1\}$ at $z_1 = (1/3,1,0,0)$.
The velocity immediately before and after $z_1$ is $v^- = (1/4,3/4,0,0)$ and $v^+ = (1/4,0,3/8, 3/8)$, respectively.
The uniform linear motion is indeed straight at $z_1$ because we have a flow over network $\mathcal{N}$ as follows: $f_{p_1 \to p_1} = 1/4$, $f_{p_2 \to p_3} = 3/8$, and $f_{p_2 \to p_4} = 3/8$.
On the other hand, the geodesic from $\bot$ to $\top$ with respect to $l_2$-metric intersects with the face at $z_2 =  (\sqrt{2} -1,1,0,0)$.
Therefore, our uniform linear motion is a different concept from $l_2$-geodesic.
\end{Example}

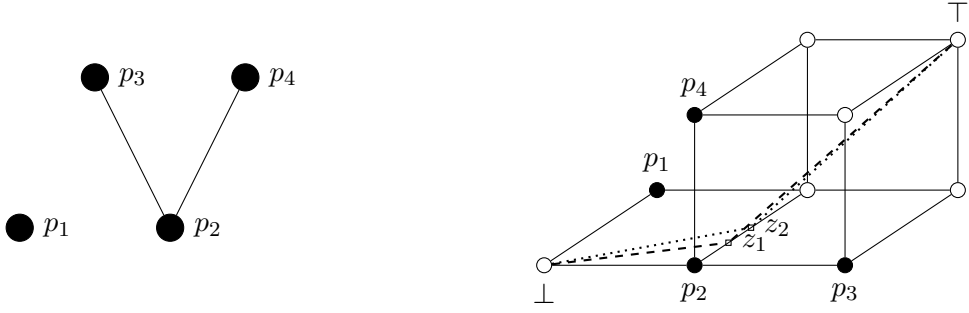
\begin{figure}
%\begin{minipage}{0.45 \hsize}
%    \centering
%    \includegraphics[width = .9 \linewidth]{figs/poset.pdf}
%\end{minipage}\begin{minipage}{0.45 \hsize}
%    \centering
%    \includegraphics[width = \linewidth]{figs/ulm.pdf}
%\end{minipage}
\begin{minipage}{0.45 \hsize}
\centering
\begin{tikzpicture}
\node[circle, draw, fill, label=east:{$p_1$}] at (0,0) (1) {};
\node[circle, draw, fill, label=east:{$p_2$}] at (2,0) (2) {};
\node[circle, draw, fill, label=east:{$p_3$}] at (1,2) (3) {};
\node[circle, draw, fill, label=east:{$p_4$}] at (3,2) (4) {};
\draw[-] (2)--(3);
\draw[-] (2)--(4);
\end{tikzpicture}
\end{minipage}
\begin{minipage}{0.45 \hsize}
\centering
\begin{tikzpicture}
\node[inner sep=2pt, circle, draw, label=south:{$\bot$}] at (0,0) (bot) {};
\node[inner sep=2pt, circle, draw, fill, label=south:{$p_2$}] at (2,0) (2) {};
\node[inner sep=2pt, circle, draw, fill, label=south:{$p_3$}] at (4,0) (3) {};
\node[inner sep=2pt, circle, draw, fill, label=north:{$p_1$}] at (1.5,1) (4) {};
\node[inner sep=2pt, circle, draw] at (3.5,1) (5) {};
\node[inner sep=2pt, circle, draw] at (5.5,1) (6) {};
\node[inner sep=2pt, circle, draw, fill, label=north:{$p_4$}] at (2,2) (7) {};
\node[inner sep=2pt, circle, draw] at (4,2) (8) {};
\node[inner sep=2pt, circle, draw] at (3.5,3) (9) {};
\node[inner sep=2pt, circle, draw, label=north:{$\top$}] at (5.5,3) (top) {};
\draw[-] (bot)--(2);
\draw[-] (2)--(3);
\draw[-] (bot)--(4);
\draw[-] (2)--(5);
\draw[-] (3)--(6);
\draw[-] (4)--(5);
\draw[-] (5)--(6);
\draw[-] (2)--(7);
\draw[-] (3)--(8);
\draw[-] (5)--(9);
\draw[-] (6)--(top);
\draw[-] (7)--(8);
\draw[-] (7)--(9);
\draw[-] (9)--(top);
\draw[-] (8)--(top);
\node[inner sep=1pt, draw, label=east:{$z_1$}] at (2.45, 0.3) (x1) {};
\node[inner sep=1pt, draw, label=east:{$z_2$}] at (2.75, 0.5) (x2) {};
\draw[-,thick,dashed] (bot)--(x1)--(top);
\draw[-,thick,dotted] (bot)--(x2)--(top);
\end{tikzpicture}
\end{minipage}
\caption{The median complex and uniform linear motion from $\bot$ to $\top$ in the example. The left figure shows the Hasse diagram of the join-irreducible elements. The right figure shows the median complex. In the right figure, join-irreducible elements are shows by black dots and the others by white ones. The dashed line indicates the uniform linear motion, where the $p_1$ coordinate of $z_1$ is $1/3$. The dotted line indicates the $l_2$-geodesic, where the $p_1$ coordinate of $z_2$ is $\sqrt{2} - 1$.}
\label{fig:ulm}
\end{figure}

The following theorem shows that for any two points in $K(\lat)$ there exists a uniform linear motion.
\begin{theorem}
\label{thm:unique-existence-uniform-linear-motion}
For any two points $x, y \in K(\lat)$, there is a unique uniform linear motion from $x$ to $y$.
\end{theorem}
Owing to the theorem, we denote by $\ulm{x}{y}(t)$ the uniform linear motion from $x$ to $y$.
In the rest of the subsection, we prove the theorem.

\paragraph{Overview of the proof}
To prove Theorem~\ref{thm:unique-existence-uniform-linear-motion}, we consider the following maximization problem.
Let $\mathcal{N}'$ be a network whose nodes are $P$ and edges consists of $p \rightarrow p'$ for $p \prec p'$.
Let $w = y - x$.
Remove all nodes with $w_p = 0$ from $\mathcal{N}'$.
If  $w_p < 0$, then we reverse the edge connecting to $p$.
This procedure is consistent by Lemma~\ref{lem:support-forms-an-ideal}.
We add bottom element $\bot$ and top element $\top$ to the following edges to network $\mathcal{N}'$:
$\bot \rightarrow p$ for all minimal $\{p \in P \mid w_p > 0\}$ and all maximal $\{p \in P \mid w_p < 0\}$; and
$p \rightarrow \top$ for all maximal $\{p \in P \mid w_p > 0\}$ and all minimal $\{ p \in P \mid w_p < 0\}$.
%Here,  all nodes and all nodes are reachable from $\top$.
All edges in the network have infinite capacity $[0,\infty)$.
In the following, a path means a directed path.
The maximization problem we consider is
\begin{align}
\label{eq:variational-max}
\begin{array}{ll}
    \underset{v}{\text{maximize}} & \sum_{p \in P} |w_p| \log(v_p)\\
    \text{subject to} 
    & v \text{ is a $\sum_{p\in P} |w_p|$-flow from $\bot$ to $\top$}.
\end{array}
\end{align}
Here, the constraint means that $v$ is decomposed as
$v = \sum_\pi f_\pi$ by non-negative $f_\pi$, where $\pi$ runs over all $\bot$ to $\top$ paths and $\sum_\pi f_\pi = \sum_{p \in P} |w_p|$.
We call $\pi$ with $f_\pi > 0$ as \emph{positive} path with respect to $v$ from $\bot$ to $\top$.
A path $\pi'$ from $p \in P$ to $p' \in P$ is said to be positive if there is a positive path $\pi$ from $\bot$ to $\top$ including $\pi'$.
For a path $\pi$, we define $\sum_{p \in \pi}$ as the summation over all node in $\pi$ without the endpoints.

Before we proceed, we explain the interpretation of the maximization problem (\ref{eq:variational-max}) and an overview of the proof.
Given unifrom linear motion $\curve$ from $x$ to $y$, we can define $v_p$ by is the $p$-th component of the speed of $\curve$ when $\curve$ pass throught  $\Cube(X)$ with $p \in X$.
By the definition of uniform linear motion, 
%$v_p$ is uniquely determined and 
$v = (v_p)_{p \in P}$ is a $\sum_{p \in P} |w_p|$-flow from $\bot$ to $\top$.
Furthermore, let  $t_p$ be the first time $c(t)_p$ pass through the hypercube $\Cube(X)$ with $p \in X$.
Then, we can see that
\begin{align}
    \label{eq:first_move_time}
    t_p = \sum_{p' \in \pi} \frac{|w_{p'}|}{v_{p'}},
\end{align}
for all positive $\bot$ to $p$ path $\pi$.
In particular, the right-hand side does not depend on the choice of $\pi$. (Lemma~\ref{lem:uniform-linear-to-variational}).
Since this property is the optimality condition for the maximization problem~(\ref{eq:variational-max}), $v$ defined above is an optimal solution  (Lemma~\ref{lem:same-derivative-of-potential-at-optimal}, Lemma~\ref{lem:uniform-linear-to-variational}).
Conversely, we can construct a uniform linear motion from the optimal solution of~(\ref{eq:variational-max}) (Lemma~\ref{lem:variational-to-uniform-linear}), which is unique due to the convexity.

\begin{lemma}
\label{lem:same-derivative-of-potential-at-optimal}
Let $v$ satisfies the constraint of the maximization problem~(\ref{eq:variational-max}.
Then, $v$ is an optimal solution of the problem if and only if
\begin{align}
    \label{eq:deriv-potential}
    \sum_{p \in \pi} \frac{|w_p|}{v_p}
\end{align}
has the same value for all positive path $\pi$.
\end{lemma}
\begin{proof}
Notice that the problem (\ref{eq:variational-max}) is a strictly convex optimization problem; hence, the optimal solution is unique.
The value (\ref{eq:deriv-potential}) is
the derivative of the objective function of (\ref{eq:variational-max}) with respect to $f_\pi$.
Thus, the optimality condition is that (\ref{eq:deriv-potential}) has the same value for all flow $\pi$ with $v_p > 0$ for all $p \in \pi$.
%By the optimality of the solution, the statement of the lemma holds.
Otherwise, by increasing a flow along a positive path with the largest value and decreasing a flow along a positive path with the smallest value, we can increase the objective value.
%The above condition is the optimality condition to the problem \eqref{eq:shortest-processing-time}.
\end{proof}

\begin{lemma}
\label{lem:uniform-linear-to-variational}
Suppose that $\curve$ is a uniform linear motion from $x$ to $y$.
Let 
\begin{align}
    v_p = \max_t \left|\frac{\dd^+ \ulm{x}{y}(t)_p}{\dd t}\right|,
\end{align}
for all $p \in P$.
Then, $v = (v_p)_{p \in P}$ is the optimal solution to the problem~(\ref{eq:variational-max}).
\end{lemma}
\begin{proof}
By the definition of the uniform linear motion, $v_p$ is the unique speed of the $p$-th component when $\ulm{x}{y}(t)$ is on a hypercube $\Cube(X)$.
Hence, $v$ satisfies $\sum_p |w_p|$-flow condition because of the straightness of the uniform linear motion.
It suffices to show that $v$ satisfies the optimal condition shown in Lemma~\ref{lem:same-derivative-of-potential-at-optimal}.
The optimality condition follows from the definition of the uniform linear motion.
Suppose to the contrary that the value (\ref{eq:deriv-potential}) is different for two positive paths $\pi_1$ and $\pi_2$.
Let $\bot = p_0 , p_1 , \dots , p_k = \top$ be the points in $\pi_1 \cap \pi_2$, where $p_i$ is ordered so that they form a subsequence of $\pi_1$.
For $\alpha = 1,2$, let $\pi_\alpha^{(j)}$ be the initial segment of $\pi_\alpha$ until $p_j$.
Take the minimum $j$ such that
\begin{align}
    \sum_{p \in \pi_1^{(j)}} \frac{|w_p|}{v_p} \neq \sum_{p \in \pi_2^{(j)}} \frac{|w_p|}{v_p}.
\end{align}
Such $j$ must exists by the assumption.
We can easily see that, for $\alpha = 1,2$, the value $t_\alpha := \sum_{p \in \pi_\alpha^{(j)}}\frac{|w_p|}{v_p}$ is the time when $\ulm{x}{y}(\alpha)_{p_{\mathrm{end}}} = y_{p_{\mathrm{end}}}$, where $p^\mathrm{end}_\alpha$ is the node $p^{\mathrm{end}}_\alpha \prec p_j$ in path $\pi_\alpha$.
Therefore, at time $\min (t_1,t_2)$, the curve $\ulm{x}{y}$ is at the surface of $K(\lat)$ and is going outside of $K(\lat)$ (in the sense of $[0,1]^P$) due to the straightness, which is a contradiction.
\end{proof}

%\COMM{SN}{notation, $p \in P$みたいに $P$の要素は$i,j$ではなくて$p$にすることにする．（下のlemの記法が辛い．）}

\begin{lemma}
\label{lem:variational-to-uniform-linear}
Let $v$ be the optimal solution of (\ref{eq:variational-max}).
Then, we can construct a uniform linear motion $\ulm{x}{y}$ from $v$.
\end{lemma}
\begin{proof}
%\COMM{SN}{(\ref{eq:variational-max})の方だけで証明を書けなくはない（(\ref{eq:shortest-processing-time})は隠ぺい可能）}
Let $v$ be the optimal solution to problem \eqref{eq:variational-max}.
Let $t_p = \max_{\pi} \sum_{p' \in \pi} |w_{p'}| / v_{p'} $, where $\pi$ runs over all positive $\bot$ to $j$ paths.
By Lemma~\ref{lem:same-derivative-of-potential-at-optimal}, the value in the max does not depend on the positive paths.
Let $X(t) = \{ p \in P : t_{p} \le t  \text{ and } t_{p'} > t \text{ for all } p' \succ p \}$.
%Then, 
We reorder $\{t_p\}_{p \in P}$ as $t_1 < t_2 < \dots < t_m$.
The latter is distinguished from the former by the subscripts $i,j,\dots$.
We define a curve $\curve$ in $\mathbb{R}^P$ by the following differential equation:
\begin{align}
    \label{eq:construct-ULM-from-variational}
    \frac{\dd^+ \curve_p(t)}{\dd t} &= \begin{cases} \mathrm{sign}(w_p) v_p & p \in X(t), \\
    0 & \text{otherwise} \ m \end{cases} \\
    \curve(0) &= x.
\end{align}

We prove that $\curve(t) \in K(\lat)$.
We first show that $\curve(t) \in [0,1]^P$.
Since there exists a positive path containing $p$, the optimality condition implies that there exists $p' \succ p$ such that $t_{p'} - t_p = |w_p| / v_p$.
Therefore, $\curve(t)_p$ moves during $|w_p| / v_p$-time with speed $v_p$, which implies that $\curve(t) \in [0,1]^P$.
We next prove that $\curve(t) \in K(\lat)$ $(t < t_k)$ for all $k$ by showing $\curve(t) \in \Cube(X(t))$ $(t < t_k)$.
It suffices to prove that $c(t_k)_{p'} = y_{p'}$ for all $k = 1,2,\dots, m$, node $p \in X(t_k) \setminus X(t_{k-1})$, and $p' \rightarrow  p$ in network $\mathcal{N}'$.
%We prove this property holds by induction on $k$.
%The base case $t = 0$ is trivial.
%We will consider the step case.
%Let $p \in X(t_k) \setminus X(t_{k-1})$.
By the definition of $t_p$ and $X(t)$, we can see that $t_p - t_{p'} \ge |w_{p'}| / v_{p'}$ for all $p' \rightarrow  p$.
Therefore, $\curve(t)_{p'}$ moves during at least $|w_{p'} / v_{p'}|$-time with speed $v_{p'}$, which implies that $\curve(t_k)_{p'} = y_{p'}$.
Here, we also used $\curve(t) \in [0,1]^P$.
In conclusion, we have shown that $\curve(t) \in K(\lat)$.

We finally prove that $\curve$ is indeed a uniform linear motion from $x$ to $y$.
The curve $\curve$ is linear on each cell by $c(t) \in \Cube(X(t))$ and the differential equation of $c$.
The curve $\curve$ is straight at face since
$v$ satisfies flow-conservation.
\end{proof}

\begin{proof}[Theorem \ref{thm:unique-existence-uniform-linear-motion}]
By Lemma \ref{lem:uniform-linear-to-variational} and Lemma \ref{lem:variational-to-uniform-linear}, the uniform linear motion corresponds to the solution of the variational problem (\ref{eq:variational-max}), which has a unique solution by strong convexity.
\end{proof}

\begin{remark}
\rm
We can prove the existence of the uniform linear motion via another optimization problem.
Consider the network $\mathcal{N}'$ defined in the maximization problem (\ref{eq:variational-max}).
In the following minimization problem, we optimize the following variables:
$t \colon P \cup \{\top,\bot\} \rightarrow \R_{\ge 0}$ and $v \colon P \cup \{\top,\bot\} \rightarrow \R$.
\begin{align}
\label{eq:shortest-processing-time}
\begin{array}{ll}
    \underset{t, v}{\text{minimize}} & t_\top \\
    \text{subject to} & t_{p'} \ge t_p + \frac{|w_{p'}|}{v_{p'}}, \quad (\forall \; p \rightarrow p') \\
    & t_\bot = 0, \\
    & v \text{ is a $\sum_p w_p$-flow from $\bot$ to $\top$}
\end{array}
\end{align}
The optimal solution of this problem corresponds to the uniform linear motion because the optimality condition for (\ref{eq:shortest-processing-time}) is equivalent to (\ref{eq:variational-max}).

This minimization problem can be interpreted as a variant of the Program Evaluation and Review Technique (PERT).
Consider the nodes of the network $\mathcal{N}'$ as tasks.
A task $p$ with $p \rightarrow p'$ must be done before task $p'$.
We have $\sum_{p \in P} |w_p|$-\emph{work capacity} in total.
Task $p$ can be done in $|w_p| / v_p$ time if $v_p$-work capacity is assigned.
Our task is to minimize the time when all task is done by appropriately distributing the work capacity.
Here, we have a constraint on the distribution of the work capacity:
A work capacity used to do task $p$ must be used to task $p'$ with  $p \rightarrow p'$, in the next.
In the above problem, $t_p$ is the time when the task $p$ is started.
\end{remark}

\subsection{Multilinear Extension}
\label{subsec:multilinear-extension}
Let $f \colon \lat \to \mathbb{R}$ be a DR-submodular function.
The \emph{multilinear extension} $F \colon K(\lat) \rightarrow  \mathbb{R}$ of $f$ is defined by
\begin{align}
    F(x) = \E_{\hat X \sim x}[ f(\hat X) ] = \sum_{X \subseteq P \text{ :ideal}}
    %f\left(\bigvee_{p \in X} p\right)
    f(X)
    \prod_{p \in X} x_{p} \prod_{p' \in P \setminus X} (1 - x_{p'}),
\end{align}
where the expectation is taken over the random ideal $\hat X$ defined by $\Prob[p \in \hat X] = x_p$ for all $p \in P$ independently and the summation $X$ is taken for all ideals of $P$.
We call $\hat X$ random ideal generated from $x$.
We omit the subscript of the expectation when $x$ is clear from the contex.
%By Lemma~\ref{lem:support-forms-an-ideal}, any summand $X$ appeared in the right-hand side is an ideal.
The multilinear extension $F(x)$ coincide with the gluing of the multilinear extension of the set submodular function on each cubes. %$\Cube(\mathrm{head}(x))$.
We note that $F$ is monotone if $f$ is monotone.

%\COMM{TM}{連続 DR を微分の単調性で定義するやつの distributive lattice 版}
%\COMM{SN}{下の補題はfaceにいるときに片側微分周りで嫌なことを引き起こすので修正してます．}
The DR-submodularity implies the inequality over the gradients as follows.
\begin{lemma}
\label{lem:continuous-DR}
Let $x, y \in K(\lat)$ with $x \le y$ and $p,q \in P$ with $p \le q$.
Then, 
\begin{align}
    \label{eq:continous-DR-from-above}
    &\frac{\partial^+ F(x)}{\partial x_p} \ge \frac{\partial^+ F(y)}{\partial y_q},\\
    &\frac{\partial^+ F(x)}{\partial x_p} \ge \frac{\partial^- F(y)}{\partial y_q},\\
    &\frac{\partial^- F(x)}{\partial x_p} \ge \frac{\partial^+ F(y)}{\partial y_q},\\
    &\frac{\partial^- F(x)}{\partial x_p} \ge \frac{\partial^- F(y)}{\partial y_q},
\end{align}
if the above one-sided partial derivative is defined.
\end{lemma}
If both $x$ and $y$ are inside of maximal hypercubes, then the above statement is equivelent to
\begin{align}
    \frac{\partial F(x)}{\partial x_p} \ge \frac{\partial F(y)}{\partial y_q}.
\end{align}
\begin{proof}
We first prove (\ref{eq:continous-DR-from-above}).
One-sided partial derivative in (\ref{eq:continous-DR-from-above}) is defined when the following conditions are satisfied: 
$x_p \neq 1$, $y_q \neq 1$, $x_{p'} = 1$ for all $p' < p$, and $y_{q'} = 1$ for all $q' < q$.
The inequality (\ref{eq:continous-DR-from-above}) is equivalent to
\begin{align}
    \E[ f(\hat X \cup \{p\}) - f(\hat X) ]
    \ge 
    \E[ f(\hat Y \cup \{q\}) - f(\hat Y) ]
\end{align}
where $\hat X$ and $\hat Y$ are the random subsets such that (1) $\Prob[p' \in \hat X] = x_{p'}$ for all $p' \in P \setminus \{p\}$, (2) $\Prob[p \in \hat X] = 0$, (3) $\Prob[q' \in \hat Y] = y_{q'}$ for all $q' \in P \setminus \{q\}$ and (4) $\Prob[q \in \hat Y] = 0$.
We prove this inequality by the coupling method. 
Let $\hat Y$ be the random ideal defined as the above.
Then, we define $\hat X$ as follows: $\hat X$ never contains $p$. 
For $p' \in P \setminus \{p\}$, 
\begin{align}
    p' \in \hat X \text{ if } p' \in \hat Y \text{ and with probability $x_{p'} / y_{p'}$}.
\end{align}
%The above probability is independent for all $p'$.
where $p' \in \hat X$ is determined independently for all $p' \in P \setminus \{p\}$.
We can see that $\Prob[p' \in \hat X] = x_{p'}$ for all $p' \in P \setminus \{ p \}$ independently and $\Prob[p \in \hat X] = 0$ as follows.
By definition, for any $p' \in P \setminus \{p, q\}$, we have $\Prob[p' \in \hat X] = \Prob[p' \in \hat Y] \cdot (x_{p'} / y_{p'}) = x_{p'}$, and $\Prob[p \in \hat X] = 0$ follows from by construction.
Therefore, the claim holds.

For any realization of $(\hat X, \hat Y)$, by the DR submodularity, we have
\begin{align}
    f(\hat X \cup \{ p \}) - f(\hat X)  
    \ge f(\hat Y \cup \{ q \}) - f(\hat Y).
\end{align}
Therefore, by taking the expectation over the joint distribution of $(\hat X, \hat Y)$, we obtain the lemma.

We can show the other inequalities by a similar argument.
\end{proof}

Recall that the multilinear extension of a set submodular function is concave along any positive direction.
This property is generalized to the distributive lattice as follows.
We define $\nabla F(x) \in \mathbb{R}^P$ by
\begin{align}
    (\nabla F(x))_p = \left\{
    \begin{array}{ll}
      \displaystyle \frac{\partial^+ F(x)}{\partial x_p} & \text{if } \displaystyle \frac{\partial^+ F(x)}{\partial x_p} \text{ is defined,}\\ 
      \displaystyle\frac{\partial^- F(x)}{\partial x_p} & \text{else if } \displaystyle \frac{\partial^- F(x)}{\partial x_p} \text{ is defined,}\\ 
      0 & \text{otherwise}.
    \end{array}\right.
\end{align}

\begin{lemma}
\label{lem:multilinear-concave-face}
Let $x, y \in K(\lat)$.
Consider any time $t$ such that $\ulm{x}{y}(t)$ is in the face of some hypercubes.
Let $h(t) = F(\ulm{x}{y}(t))$.
\begin{comment}
Let
\begin{align}
     &\left.  \frac{d F(\ulm{x}{y}(t))}{d t} \right|_{s + 0} = \lim_{\epsilon \to 0^+} \frac{F(\ulm{x}{y}(t+\epsilon)) - F(\ulm{x}{y}(t))}{\epsilon}, \\
     &\left.  \frac{d F(\ulm{x}{y}(t))}{d t} \right|_{s - 0} = \lim_{\epsilon \to 0^+} \frac{F(\ulm{x}{y}(t)) - F(\ulm{x}{y}(t- \epsilon))}{\epsilon},
\end{align}
where the above limit is one-sided, i.e., $\epsilon$ approaches to $0$ from above.
\end{comment}
Then,
\begin{align}
    \frac{\dd^+ h(t)}{\dd t} \le \frac{\dd^- h(t)}{\dd t}.
\end{align}
\end{lemma}
\begin{proof}
We define $v^+$ and $v^-$ as in (\ref{eq:velocity-at-face}).
Let $f$ be the flow over network $\mathcal{N}$ in the definition of straightness at the face where $\ulm{x}{y}(t)$ lies on.
%We use the same notation as the definition of the straightness (around (\ref{eq:velocity-at-face}).
We decompose $f$ as $f = \sum_{p,p'} f_{pp'}$ where $p$ is in the  support of $v^+$ and $p'$ is in the support of $v^-$.
%Notice that we can take the decomposition so that $f_{pp'} > 0$ always connects $(v_+)_i > 0$ and $(v_-)_i > 0$, and $f_{ij} < 0$ connects $(v_+)_i < 0$ and $(v_-)_i < 0$ since $X$ and $X'$ are antichains.
Let $F_+$ be the set of the pairs $(p,p')$ with  $f_{pp'} > 0$ and $v^+_p > 0$.
Also, let $F_-$ be the set of $(p,p')$ with $f_{pp'} > 0$ and  $v^+_p < 0$.
%Let $z = \ulm{x}{y}(t)$.
By Lemma \ref{lem:continuous-DR},
\begin{align}
   \label{eq:conti-DR-proof-1}
   \frac{\partial^+ F(\ulm{x}{y}(t))}{\partial x_p} \le \frac{\partial^-F(\ulm{x}{y}(t))}{\partial x_{p'}},
\end{align}
for all $(p,p') \in F_+$ (we remark that $p \succ p'$ in this case), and 
\begin{align}
       \label{eq:conti-DR-proof-2}
       \frac{\partial^- F(\ulm{x}{y}(t))}{\partial x_p} \ge \frac{\partial^+F(\ulm{x}{y}(t))}{\partial x_{p'}},
\end{align}
for all $(p,p') \in F_-$ ($p \prec p'$ in this case).
By (\ref{eq:conti-DR-proof-1}), (\ref{eq:conti-DR-proof-2}) and the definition of straightness of uniform linear motion at face
\begin{align}
   (\nabla F(\ulm{x}{y}(t+0)))_p f_{p p'} \mathrm{sign}(v^+_p) \le (\nabla F(\ulm{x}{y}(t)))_{p'}f_{p p'} \mathrm{sign}(v^-_{p'}),
\end{align}
for all $(p,p') \in F_+ \cup F_-$.
%where $\mathrm{sign}(p,p') = 1$ for $(p,p') \in F_+$ and $\mathrm{sign}(p,p') = 0$ otherwise.
By the definition of the straightness, $|v^+_p| = \sum_{p'} f_{p p'}$, where $p'$ runs over all $p'$ with $(p,p') \in F_+ \cup F_-$ and $|v^-_{p'}| = \sum_p f_{p p'}$, where $p$ runs over all $p$ with $(p,p') \in F_+ \cup F_-$.
By using these equations, we have
\begin{align}
   \frac{\dd^+ h(t)}{\dd t} &= \langle \nabla F(\ulm{x}{y}(t)), v^+ \rangle\\
   &= \sum_{p} (\nabla F(\ulm{x}{y}(t)))_p  \sum_{p'} \mathrm{sign}(v^+_{p}) f_{p p'}\\
   &=  \sum_{p,p'} (\nabla F(\ulm{x}{y}(t)))_p  \mathrm{sign}(v^+_p) f_{p p'}\\
   &\le \sum_{p,p'} (\nabla F(\ulm{x}{y}(t)))_{p'} \mathrm{sign}(v^-_{p'})f_{p p'}\\
   &= \langle \nabla F(\ulm{x}{y}(t)), v^-\rangle\\
   &= \frac{\dd^-h(t)}{\dd t},
\end{align}
where the summations are taken for all $(p,p') \in F_+ \cup F_-$.
\end{proof}

\begin{theorem}
\label{thm:multilinear-concave}
For any $x, y \in K(\lat)$ with $x \le y$, the function $h(t) := F(\ulm{x}{y}(t))$ is concave in $t \in [0,1]$.
\end{theorem}
\begin{proof}
On each hypercube, $F(\ulm{x}{y}(t))$ is concave since $\ulm{x}{y}(t)$ is line in a positive direction and $F$ restricted on this hypercube is the multilinear extension over Boolean lattice.
Thus, it suffices to show that $h$ is concave at faces, which have already shown in Lemma~\ref{lem:multilinear-concave-face}.
\end{proof}

\begin{comment}
\begin{lemma}[Concave Upper Bound]
\label{lem:concave-upper-bound}
For any $\epsilon > 0$,
\begin{align}
    F(c_{x,y}(\epsilon)) - F(c_{x,y}(0)) \le \langle \dot{c}_{x,y}(0), \nabla F(x) \rangle.
\end{align}
\end{lemma}

\begin{lemma}[Quadratic Lower Bound]
\label{lem:quadratic-lower-bound}
Assume that $f(p) \le 1$ for all $p \in P$.
For any $\epsilon > 0$,
\begin{align}
    F(c_{x,y}(\epsilon)) - F(c_{x,y}(0)) \ge \langle \dot{c}_{x,y}(0), \nabla F(x) \rangle - \epsilon^2 P^2.
\end{align}
\end{lemma}
\begin{proof}
This is the Taylor theorem, where the second term is the remainder term evaluated by
\begin{align}
    R = \frac{\epsilon^2 }{2} \dot{c}_{x,y}(\xi)^\top \nabla^2 F(c_{x,y}(\xi)) \dot{c}_{x,y}(\xi)
\end{align}
for some $\xi \in [0, \epsilon]$.
By the definition of the uniform linear motion, $\| \dot{c}_{x,y}(\xi) \|_1 \le \| x - y \|_1 \le |C|$.
Also, each entry of $\nabla^2 F(c_{x,y}(\xi)$ is non-positive and lower bounded by $-2$.
Therefore, $R \ge -\epsilon^2 P^2$.
\COMM{TM}{二階微分の計数評価があやしい}
\end{proof}
\end{comment}

%\COMM{TM}{
%convex な制約の例を複数作らないといけない．
%\begin{itemize}
%    \item ナップサック制約？
%    \item マトロイド制約？
%\end{itemize}
%}

\section{Continous Greedy Algorithm over Median Complex}
\label{sec:maximization}

%In this section, we consider multiple knapsack constraint $G_\lambda(x) = \langle c_\lambda, x \rangle$ ($\lambda \in \Lambda$), where $c_\lambda$ is order-consistent.
In this section, we propose an algorithm that finds $1 - 1/e$-approximation solution for the continuous relaxation of the DR-submodular maximization problem under multiple knapsack constraints. 
%\COMM{SN}{greedyの部分を各cubeでのLPにする．}
%We denote the multilinear extension of $c_\lambda$ by the same notation.
The multilinear extension of $c_\lambda$ is $C_\lambda(x) = \sum_{p \in P} c_\lambda(p) x_p$ for $x \in K(\lat)$.
Therefore, the continuous relaxation of the feasible region of the multiple knapsack constraint is 
\begin{align}
    \Omega = \{ x \in K(\lat) \mid C_\lambda(x) \le b_\lambda \text{ for all } \lambda \in \Lambda\}.
\end{align}
Let $\epsilon > 0$.
%To simplify the notation, we write $c_\lambda(p)$ as $c_{p,\lambda}$.
The \emph{epsilon upper neighborhood} $N_{\epsilon}^+(x)$ of $x \in K(\lat)$ associated with knapsack constraint $(c_\lambda, b_\lambda)_{\lambda \in \Lambda}$ is defined by
\begin{align}
    N_{\epsilon}^+(x) = \{y \in K(\lat) \mid y \ge x,\quad C_\lambda(y) \le C_\lambda(x) + \epsilon b_\lambda  \; (\forall \lambda \in \Lambda)\}.
\end{align}
For $x \in k(\lat)$, 
%let $C(x) = \mathrm{head}(x)$ if the head hypercube $\mathrm{head}(x)$ is not the face of some maximal cubes.
%If $x$ lies on the face of some maximal cubes, then let $C(x)$ be the uppest one, i.e.,  
let $\Cube(x)$ be the cube corresponding to the antichain that consists of the maximal elements of $\mathrm{supp}(x) \cup \{p \in P \mid x_{p'} = 1 \text{ for all } p' \prec p\}$.
%When $x$ is not on the face of some maximal hypercube, $\Cube(x) = \mathrm{head}(x)$.

\paragraph{Overview of the algorithm}
Algorithm~\ref{alg:monotone-test} outputs an approximate solution by updating $x^k$ $(k = 0,1,\dots, \floor{1/\epsilon})$ as follows.
At each step, $x^k$ moves sufficiently small, i.e., $x^{k+1} \in N_\epsilon^+(x)$.
This property guarantees the feasibility of the output $x^{\lfloor 1/\epsilon \rfloor}$.
When $x^k$ is far from the surface of the maximal hypercube where $x^k$ lies on, i.e., $N_\epsilon^*(x^k) \subseteq \Cube(x^k)$, then we solve the following linear programming in $\Cube(x^k)$ to update $x^k$: let $x^{k+1} \gets \mathrm{argmax}_{y \in N_\epsilon^+(x)}. \langle y - x^k, \nabla F(x^k) \rangle$.
This LP is easily solved since $\Cube(x^k)$ is a Euclidean space.
In this case, we can assure that the objective value increases sufficiently (Equation~(\ref{eq:update-objective-value})) by using the concavity of the $F$ for positive directions (Theorem~\ref{thm:multilinear-concave}).
Namely, we prove that the direction $\frac{\dd^+ \ulm{x^k}{x^k \lor x^*}(0)}{\dd t}$ is feasible and have sufficiently large innerproduct with $\nabla F(x^k)$.
Otherwise, i.e. $N_\epsilon^+(x^k) \not \subseteq \Cube(x^k)$, the solution of the LP might be the outside of $\Cube(x^k)$.
To avoid this difficulty, we just set $x^{k+1}$ to some point in an upper hypercube.
The only bound we have is a trivial one: $f(x^{k+1}) \ge f(x^k)$.
However, this is not a problem because such cases happen at most $|P|$-times.

\begin{algorithm}[tb]
\caption{Continuous greedy algorithm for median complex under multiple knapsack constraints}
\label{alg:monotone-test}
\begin{algorithmic}[1]
\State{$x^0 = \bot$}
\For{$k = 0, 1, \dots, \floor{1/\epsilon} - 1$}
\If{$N_\epsilon^*(x^k) \subseteq \Cube(x^k)$}
\State{Find $y \in N_\epsilon^*(x^k)$ that maximizes
$\langle y-x^k, \nabla F(x^k) \rangle $.}
\State{$x^{k+1} \gets y$.}
\Else
\State{Set $x^{k+1}$ to an arbitrary point in $N_\epsilon^*(x) \setminus \Cube(x^k)$. }
\EndIf
\EndFor
\State{Return $x^{\lfloor 1 / \epsilon \rfloor}$.}
\end{algorithmic}
\end{algorithm}

\begin{theorem}
\label{thm:conti-greedy-for-knapsack}
Let $f$ be a monotone DR-submodular function and $F$ its multilinear extension.
Then, Algorithm \ref{alg:monotone-test} outputs a feasible solution $x \in \Omega$ satisfying
\begin{align}
    F(x) \ge \left(1 - (1 - \epsilon)^{\lfloor 1/\epsilon \rfloor - |P|}\right) F(x^*),
\end{align}
where $x^* = \mathrm{argmax}_{x \in \Omega} F(x)$.
\end{theorem}
We note that the approximation ratio converges to $1 - e^{-1}$ as $\epsilon \to 0$.
\begin{proof}
By the definition of the upper epsilon neighborhood, the obtained solution is feasible. 
Therefore, it suffices to prove the approximation factor.
%If $N_\epsilon^+(x) \subseteq C(x^k)$, where $C(x^k)$ is the maximal cube containing $x^k$, then the greedy property implies that
If  $N_\epsilon^+(x) \subseteq C(x^k)$, we have
\begin{align}
\label{eq:update-objective-value}
    F(x^{k+1}) - F(x^k) 
    &= F(y) - F(x^k) \\
    &\ge^{(*)} \langle y - x, \nabla F(x^k) \rangle - O(\epsilon^2) \\
    &\ge^{(**)}  \langle \ulm{x^k}{x^k\lor x^*}(\epsilon) - x, \nabla F(x^k) \rangle - O(\epsilon^2) \\
   % &=  \epsilon \left \langle \frac{\dd^+ \ulm{x^k}{x^* \lor x^k}(0)}{\dd t}, \nabla F(x^k) \right \rangle - O(\epsilon^2) \\
    &\ge^{(*)} F(\ulm{x^k}{x^k \lor x^*}(\epsilon)) - F(x^k)  - O(\epsilon^2)\\
    &\ge^{(***)} \epsilon \left( F(x^k \lor x^*) - F(x^k) \right) - O(\epsilon^2) \\
    &\ge \epsilon \left( F(x^*) - F(x^k) \right) - O(\epsilon^2).
\end{align}
In (*), we used the Taylor's theorem.
In (**), we used the greedy property of the algorithm and the fact that $\ulm{x^k}{x\lor x^*}(\epsilon) \in N_\epsilon^+(x_k)$.
%by the same argument of the proof of Lemma \ref{lem:convex-combination-of-joins}.
%\COMM{SN}{$c$ が色々なところで被っている．$c$ (curve, knapsack), $C$ (multilinear of knapsack, cube)}
The above statement holds
since order-consistency of $c_\lambda$ implies that
\begin{align}
    C_\lambda(\ulm{x^k}{x^k\lor x^*}(\epsilon)) - C_\lambda(x^k) \le \epsilon (C_\lambda(x^k \lor x^*) - C_\lambda(x^k)) \le  \epsilon (C_\lambda(x^k) + C_\lambda(x^*) - C_\lambda(x^k)) = \epsilon C_\lambda(x^*) \le \epsilon b_\lambda,
\end{align}
where we used $C_\lambda(x \lor y) \le C_\lambda(x) + C_\lambda(y)$ in the second last inequality.
In (***), we used Theorem~\ref{thm:multilinear-concave}.
%and the fact that  $\ulm{x^k}{x^* \lor x^k}(t)$ is linear for $t \in [0, \epsilon]$ since $N_\epsilon^+(x) \subseteq C(x^k)$.
In the last inequality, we used the monotonicity of $F$.

If  $N_\epsilon^+(x) \not \subseteq C(x^k)$, then the monotonicity of $F$ implies that $F(x^{k+1}) \ge F(x^k)$.
We note that $N_\epsilon^+(x) \subseteq C(x^k)$ holds except $|P|$-times in the execution of the Algorithm.
Therefore, by solving this recursive inequality, we have
\begin{align}
    F\left(x^{\lfloor 1/\epsilon \rfloor}\right) \ge \left(1 - (1 - \epsilon)^{\lfloor 1/\epsilon \rfloor - |P|}\right) F(x^*) - O(\epsilon).
\end{align}
\end{proof}

\section{Rounding for Multiple Knapsack Constraints}
\label{sec:rounding}

We propose an approximation algorithm for monotone DR-submodular function maximization under multiple knapsack constraint by combining continuous greedy algorithm over median complex (Section \ref{sec:maximization}) and generalization of rounding technique~\cite{kulik2009maximizing}.

\subsection{Preliminary}
We note that the knapsack constraint $c_\lambda$ is also defined for the subset $X$ (not only for ideal) of $P$ by $c_\lambda(X) = \sum_{p \in X} c_{\lambda}(p)$.
Let $\epsilon > 0$ be some sufficiently small ($<1/2$) parameter (different from $\epsilon$ in the previous section).
An element $p \in P$ is said to be \emph{small in dimension $\lambda$} 
if $c_{\lambda}(p)  \le \epsilon^4 b_\lambda$.
An element is \emph{small} if it is small in all dimension $\lambda \in \Lambda$.
An element is big if it is not small.

We will use the following two \emph{residual problems} with respect to an ideal $T$.
\begin{itemize}  
    \item \emph{Value residual problem} with integer parameter $h$. In this problem, the underlying poset $P'$ consists of $p \in P \setminus T$ such that $f_T(I_p) \le f(T) / h$.
    \item \emph{Cost residual problem}. In this problem, the underlying poset $P'$ consists of all small elements in $P \subseteq T$.
\end{itemize}.
The objective function of the residual problems is $f_T(X)$ for an ideal $X$ of $P'$, and the knapsack constraint is $\{X \subseteq P' \mid c_\lambda(X) \le b_\lambda - c_\lambda(T) \}_{\lambda \in \Lambda}$.
%We call these constraints \emph{residual knapsacks}.
%, where %$c_\lambda(X;T) = c(X \cup T) - c(T)$ and 
%$b_\lambda(T) = b_\lambda - c_\lambda(T)$.
In the value residual problem, the underlying poset $P'$ is an ideal of $P \setminus T$ due to the monotonicity of the objective function.
The same thing holds for the cost residual problem due to the order-consistency of $c_\lambda$.
%Furthermore, in the cost residual problem, $P'$ does not contain an element $p$ greater than a big element in $T$ due to the DR-supermodularity of the cost.

In the following, we use uniform linear motion both for the original problem and the residual problem.
To clarify the considering poset, we use superscript $P'$ as $\ulm{x}{y}^{P'}(t)$ for the uniform linear motion over the median complex arose from $P'$.

\subsection{Overview of the Algorithm}
We later propose an algorithm $\mathcal{A}$ (Algorithm~\ref{alg:rounding}) by combining continuous greedy algorithm and rounding technique~\cite{kulik2009maximizing} in Section~\ref{subsec:rounding}.
Consider an instance $\mathcal{I}' = (P', f', \{(c_\lambda', b_\lambda')\}_{\lambda \in \Lambda})$ of DR-submodular maximization problem under multiple knapsack constraints, where $P'$ is the underlyning poset, $f'$ is the objective function, and $(c_\lambda',b_\lambda')$ is the knapsack constraint.
The algorithm $\mathcal{A}$ outputs a random feasible solution $\hat D \subseteq P'$ of $I$ satisfying
\begin{align}
    \label{eq:approx-gurantee-rounding}
    \E[f(\hat D)] \ge (1 - \Theta(\epsilon))(1 - e^{-1})f(\OPT) - |\Lambda| \epsilon^3 M,
\end{align}
where $\OPT$ is the optimal solution of the instance $\mathcal{I}'$ and $M = \max_{p' \in P'} f(I_{p'})$.
In summary, $\mathcal{A}$ has almost $1 - e^{-1}$-approximation guarantee if the profit $f(I_{p'})$ is sufficiently small for all $p' \in P'$.
%The details on $\mathcal{A}$ is explained in Section~\ref{subsec:rounding}.

To achieve $1 - e^{-1}$ approximation guarantee, we employ a technique called partial enumeration.
In the main algorithm (Algorithm~\ref{alg:partial-enumeration}), we enumerate the sets $X \subseteq P$ consists of at most  $\lceil e d \epsilon^3 \rceil$ elements.
After enumeration, we use algorithm $\mathcal{A}$ to solve the value residual problem for $\bigcup_{p \in X} I_p$.
At some iteration, $X$ contains $p \in \OPT$ with high profit and $\mathcal{A}$ approximately finds a good approximation of the other elements since the profit of such elements are low and $M$ is small.

\subsection{Partial enumeration}
\label{subsec:partial-enumeration}
We propose the main algorithm that uses $\mathcal{A}$ explained in Section~\ref{subsec:rounding} as a subroutine.
We note that the runtime of the proposed Algorithm~\ref{alg:partial-enumeration} is polynomial in $|P|$.

\begin{algorithm}
\caption{Partial enumeration}
\label{alg:partial-enumeration}
\begin{algorithmic}[1]
\State{$\hat S \gets \emptyset$.}
%\For{all ideals $T \subseteq P$ with $|T| \le \lceil e d \cdot \epsilon^{-3} \rceil$:}
\For{each \emph{subset} $X \subseteq P$ with $|X| \le \lceil e d  \epsilon^{-3} \rceil$:}
\State{$T \gets \bigcup_{p \in X} I_p$. (If $T$ is not feasible, then skip to the next iteration.)}
\State{Let $\hat D$ be the output of $\mathcal{A}$ for the value residual problem of $T$ with respect to parameter $|X|$.}
\If{$f(\hat D \cup T) > f(\hat A)$:}
\State{$\hat S \gets \hat  D \cup T$.}
\EndIf{}
%\State{$x((i+i)\mathrm{d}t) \gets c_{x(i \mathrm{d}t), p}(\mathrm{d} t)$}
\EndFor
\State{\Return $\hat S$.}
\end{algorithmic}
\end{algorithm}

\begin{theorem}
Let $\hat S$ be the output of Algorithm~\ref{alg:partial-enumeration}.
Then, $\hat S$ is feasible and
\begin{align}
    \E[f(\hat S)] \ge  (1 - \Theta(\epsilon))(1-e^{-1})f(\OPT),
\end{align}
where $\OPT$ is an optimal solution.
%Algorithm~\ref{alg:partial-enumeration} has-approximation gurantee.
\end{theorem}
\begin{proof}
%In this proof, we use the same notation as \cite{kulik2009maximizing}.
%The proof of Theorem 2.1 in \cite{kulik2009maximizing} holds for distributive lattices with a little modification.

%The modification is the following
The feasibility of $\hat S$ is trivial.
We prove the expected approximation guarantee.
Let $\OPT  = \{p_1,p_2, \dots, p_k \}$, where
we choose the order of $\{p_l\}_{l = 1,2,\dots, k}$ greedily.
Precisely, we choose $p_l$ recursively as follows.
Let $K_0 = \emptyset$.
For $l = 1, 2, \dots, k$, we iterate the following.
\begin{enumerate}
    \item If $K_{l-1} = \OPT$, then take $p_{l'}$ $(l' = l,l+1,\dots, k)$ arbitrary and finish the procedure. 
    \item choose $p_l \in \OPT$ as $p_l = \argmax_{p \in \OPT \setminus K_{l-1}} f_{K_{l-1}}(I_{p})$.
    \item $K_{l} = K_{l-1} \cup I_{p_l}$.
\end{enumerate}
By this construction, we can show the following.
\begin{itemize}
    \item $K_l$ is an ideal.
    \item If $j > l$, then $p_j \in K_l$ or $f_{K_{l-1}}(I_{p_j}) \le f_{K_{l-1}}(I_{p_l})$. Here, we used the directional DR-submodularity and the greedy construction.
    By using this property recursively on $l$, we have
    \begin{align}
        \label{eq:OPT-low-profit}
        f_{K_{l-1}}(I_{p_j}) \le f(K_l) / l.
    \end{align}
\end{itemize}

Let $h = \lceil e d \epsilon^{-3} \rceil$.
If $k \le h$, then $\OPT$ is chosen as $T$ at some iteration and the output is optimal.
We analyze the other case: $k > h$.
In this case, we focus on the iteration where $K_h$ is chosen as $T$ and prove the approximation guarantee of $S \cup T$ at step 5 since $\mathcal{A}$ outputs the maximum $S \cup T$.
At this iteration, $|X| = h$.
Let $\alpha = f(T)/f(\OPT)$.
The optimal solution of the value residual problem is $(1-\alpha)f(\OPT)$
since the underlying poset $P'$ of the value residual problem contains $\OPT \setminus K_h$ due to (\ref{eq:OPT-low-profit}).
Hence, by the approximation gurantee of $\mathcal{A}$ (Theorem~\ref{thm:rounding_knapsack}),
\begin{align}
   \E[f_T(\hat D)] \ge (1 - \Theta(\epsilon))(1 - e^{-1})(1 - \alpha) f(\OPT) - |\Lambda| \epsilon^{-3}M,
\end{align}
where $M = \max_{p' \in P'} f_T(I_{p'})$.
%and $P'$ is the underlying set for the value residual problem.
By the definition of the value residual problem, we know that 
\begin{align}
    M \le \frac{f(T)}{h} \le h^{-1} \alpha f(\OPT).
\end{align}
Therefore,
\begin{align}
   \E[f(\hat S)] &= \E[ f(\hat D \cup T)]\\
   &= f(T) + \E[f_T(\hat D)]\\
    &\ge \alpha f(\OPT) +  (1 - \Theta(\epsilon))(1 - e^{-1})(1 - \alpha) f(\OPT) - |\Lambda| \epsilon^{-3} \cdot \frac{\epsilon^3}{|\Lambda| e} f(\OPT)\\
    &\ge (1 - \Theta(\epsilon))(1 - e^{-1}) f(\OPT).
\end{align}
\end{proof}

\subsection{Rounding}
\label{subsec:rounding}
We next explain the details on the subroutine $\mathcal{A}$.
The goal of this section to prove the following approximation guarantee (Theorem~\ref{thm:rounding_knapsack}) of $\mathcal{A}$ for any instances $\mathcal{I} = (P', f', \{(c_\lambda', b_\lambda)\}_{\lambda \in \Lambda})$ (possibly different from the original problem over $P$).
%Precisely, given a (possibly restricted) underlying set $P' \subseteq P$ and the corresponding knapsack constraints, $\mathcal{A}$ returns a random feasible solution $S$ such that 
To simplify the notation, we consider the case $\mathcal{I} = (P, f, \{(c_\lambda,b_\lambda)\}_{\lambda \in \Lambda})$ without loss of any generality.
Let $\OPT$ be an optimal solution of this problem.
\begin{theorem}
\label{thm:rounding_knapsack}
Let $\hat D$ be the output of $\mathcal{A}$ (Algorithm~\ref{alg:rounding}).
Then, $\hat D$ is feasible and
\begin{align}
    \E[f(\hat D)]  \ge (1 - \Theta(\epsilon))(1 - e^{-1})f(\OPT) - |\Lambda|\epsilon^{-3}M,
\end{align}
where $M = \max_{p \in P} f(I_p)$.
\end{theorem}
%We prove this theorem by using the following lemmas.
%First of all, we define the algorithm $\mathcal{A}$.

\paragraph{Overview of the algorithm}
%Before we give the definition of $\mathcal{A}$, 
Let us explain the overview of the rounding (Algorithm~\ref{alg:rounding}) of the fractional solution $x \in K(\lat)$ given by  Algorithm~\ref{alg:monotone-test}.
The difficulty in rounding is that the random ideal $\hat X$ generated from $x \in K(\lat)$ might violate the constraint even though $x$ is feasible.
To avoid the difficulty, we enumerate subsets $T$ of the big elements (step 1) and then obtain $x \in K(\lat)$ by solving the cost residual problem (step 2) with respect to $T$.
%Let the the resultant random ideal $\hat D = \hat X \cup T$.
Since $p \in \hat X$ is small, the variance of $c_\lambda(\hat X)$ is small.
Owing to the small variance, random ideal $\hat D := T \cup \hat X^{1-\epsilon}$, where $\hat X^{1 - \epsilon}$ is generated from truncated $x^{1-\epsilon} = \ulm{\bot}{x}(1-\epsilon)$, satisfies the constraint with high probability (from step 3 to step 5; Lemma~\ref{lem:kulik2.2}) if $T$ do not occupy the most of the knapsack.
%, equivalently, the residual knapsack capacity $b_\lambda - c_\lambda(T)$ is sufficiently large.
Also, we can guarantee that $\E[f(\hat D)]$ is sufficiently large (Lemma~\ref{lem:kulik2.1} and Lemma~\ref{lem:bound-step-5}).
If $T$ occupies the most of the knapsack, we can make $\hat D$ feasible by
removing big elements from $\hat D$ until $\hat D$ becomes feasible (step 6).
The additive error $|\Lambda|\epsilon^3M$ comes from this removal procedure.
When big elements are removed from $\hat D$, the resulting set might not be an ideal.
Thus, we convert this set to an ideal by Algorithm~\ref{alg:pushdown} (step 7).
We can guarantee that the solution is kept feasible and does not become worse so much by these modifications (see the proof of Theorem~\ref{thm:rounding_knapsack} below).

%\paragraph{The procedure of the algorithm}
%We give the procedure of the algorithm $\mathcal{A}$ as follows:

\begin{algorithm}[tb]
\caption{Rounding algorithm $\mathcal{A}$}
\label{alg:rounding}
\begin{algorithmic}[1]
    \State{Consider all feasible solutions of the form $T = \bigcup_p I_{p}$, where $p$ runs over some subset of the big elements.
        For each $T$, we do the followings steps.
        Let $T_\lambda' \subseteq T$ be the set of the element that is big in dimention $\lambda$.
        %(Here, $T_\lambda'$ is not necessarily an ideal.)
        Let $T_\lambda = \bigcup_{p \in T_\lambda'} I_p \subseteq T$.}
    \State{Apply Algorithm~\ref{alg:monotone-test} to the cost residual problem with respect to $T$. Let $\bar{x}$ be the solution and $P'$ be the underlying poset of the cost residual problem.}
    \State{Let $\bar{x}^{1-\epsilon}  = \ulm{\perp}{\bar{x}}^{P'}(1 - \epsilon)$ and $x_T^{1-\epsilon} = \ulm{\perp}{x_T}^{P}(1 - \epsilon)$, where $x_T$ is the integer point in $K(\lat(P))$ corresponding to $T$.
        Let $\hat D_1$ be a random ideal generated from $\bar{x}^{1-\epsilon} \in K(I(P'))$ and $\hat D_2$ from $x_T^{1-\epsilon} \in K(I(P))$.} %(see also the definition of the multilinear extension (Section \ref{subsec:multilinear-extension})).
    \State{Let $c_\lambda^g = c_\lambda(T_\lambda)$ and $\bar{b}_\lambda = b_\lambda - c_\lambda^g$.}
    \State{If one of the following holds for some $\lambda$, then let $\hat D = \emptyset$ and skip to the step 8.
        \begin{itemize}
            \item $\bar{b_\lambda} > \epsilon b_\lambda$ and $c_\lambda(\hat D) > b_\lambda$.
            \item $\bar{b_\lambda} \le \epsilon b_\lambda$ and $c_\lambda(\hat D \setminus T_\lambda) > \epsilon b_\lambda + \bar{b}_\lambda$.
        \end{itemize}
        }
    \State{For all $\lambda$ such that $\bar{b}_\lambda \le \epsilon b_\lambda$, remove the element of $\hat D$ from $\hat D_2$ until $c_\lambda(\hat D_1 \cup \hat D_2) \le b_\lambda$ for all $\lambda$ as follows.
    %(Here, $\hat D = \hat D_1 \cup \hat D_2$ is not necessarily an ideal.)
    \begin{itemize}
        \item Choose arbitrary element $p \in T_\lambda'$, and let $T_\lambda' \gets T_\lambda' \setminus \{p\}$ and $T_\lambda \gets \bigcup_{\alpha \in T_\lambda' \setminus \{p\}} I_\alpha$. 
        \item  $\hat D_2 \gets \bigcup_\lambda T_\lambda$.
    \end{itemize}
    }
    \State{Convert $\hat D$ into an ideal by Algorithm~\ref{alg:pushdown}.}
    \State{Return $\hat D$ with maximum $f(\hat D)$ with respect to the iteration on $T$.}
\end{algorithmic}
\end{algorithm}

\begin{algorithm}[tb]
\caption{Push down}
\label{alg:pushdown}
\hspace*{\algorithmicindent} \textbf{Input}: A set $\hat D_1 \subseteq P$ and an ideal $\hat D_2 \subseteq P$ with $\hat D_1 \cap \hat D_2 = \emptyset$.
\begin{algorithmic}[1]
\State{Take some topological order $\hat D_1 = \{p_1, p_2, \dots, p_{|\hat D_1|}\}$  (A larger element appears later).}
\State{$\hat D \gets \hat D_2$.}
\For{$t = 1,2,\dots, |\hat D_1|$:}
%\State{Let $l(t)$ be the $t$-th minimum number of $\{j \in \mathbb{N} | p_j \in D_1\}$.}
\State{Take some $p'_{t} \le p_{t}$ with $p'_{t} \in \adm{\hat D}$.}
\State{$\hat D \gets\hat  D \cup \{p'_t\}$.}
\EndFor
\State{\Return $\hat D$.}
\end{algorithmic}
\end{algorithm}

In the following, we prove Theorem~\ref{thm:rounding_knapsack}.
It suffices to prove the approximation guarantee of $F$ at the iteration when $T$ contains all big elements in $\OPT$ since
the algorithm takes the maximum at the output step.
%Let $W = f(D)$, where $D$ is considered after step 3. 
Let $\hat D^{(t)}$ be $\hat D$ after step $t$.
We define $\hat D^{(t)}_1$ and $\hat D^{(t)}_2$ in the same way.
%\COMM{TM}{記号 $W$ 入れる必要ある？}
%\COMM{SN}{必要ないです．元論文に合わせて書いてました．記号はこっちで再定義してしまってよいと思います．（証明もそのまま移植ではなくなってしまったので）}
\begin{lemma}
\label{lem:kulik2.1}
The set $\hat D^{(3)}$ is an ideal and
\begin{align}
    \E[f(\hat D^{(3)})] \ge (1 - \Theta(\epsilon)) (1 - e^{-1})f(\OPT).
\end{align}
\end{lemma}
\begin{proof}
To simplify the notation, we omit the superscript $(3)$ in this proof.
We first show that the set $\hat D$ is an ideal.
Since $\epsilon$ is sufficiently small, $T / \hat D_2$ consists of maximal elements of $T$.
Therefore, $\hat D_1$ does not contain elements greater than $T \setminus \hat D_2$ since such elements are big.
Hence, $\hat D$ is an ideal.

We next prove the inequality.
The DR-submodularity of $f$ and $\hat D_2 \subseteq T$ imply that
\begin{align}
    \E[f(\hat D)] = \E[f(\hat D_2) + f_{\hat D_2}(\hat D_1)] \ge \E[f(\hat D_2)] + \E[f_T(\hat D_1)].
\end{align}
%Let $F$ be the multilinear extension of $f$.
Let $F_T$ be the multilinear extension of $f_T$ over $K(I(P'))$l
By the concavity of the multilinear extension $F$ along positive directions (Theorem~\ref{thm:multilinear-concave}), we have
\begin{align}
    \E[f(\hat D_2)] = F(x_T^{1 - \epsilon}) \ge (1 - \epsilon) F(x_T),
\end{align}
and
\begin{align}
    \E[f_T(\hat D_1)] = F_T(\bar{x}^{1-\epsilon}) \ge (1 - \epsilon) F_T(\bar{x}).
\end{align}
Since $T$ contains all big elements in $\OPT$, the set $\OPT \setminus T$ is an feasible solution for the cost residual problem.
Therefore, Theorem~\ref{thm:conti-greedy-for-knapsack} implies that
\begin{align}
    F_T(\bar{x}) 
    %\ge  (1 - e^{-1} - \epsilon) \max_x F_T(x) 
    \ge (1 - e^{-1} - \epsilon) f_T(\OPT \setminus T),
\end{align}
where $x$ in the maximization runs over feasible $x \in K(I(p'))$.
%where $x^*$ is the optimal solution of the continuous relaxation of the cost residual problem: $x^* = \argmax_{x \in \Omega'} F_T(x)$ and $\Omega'$ is the residual knapsack constraint.
%Here, by the choice of $T$ (containing all big element in $\OPT$), the above $\OPT \setminus T$ is a feasible solution for the cost residual problem in step 2. 
By combining these inequalities and inequalities, we have the desired inequality.
\end{proof}

\begin{lemma}
  \label{lem:kulik2.2}
   The ideal $\hat D^{(3)}$ satisfies the conditions of step 5 for some $\lambda$ with probability at most $|\Lambda| \epsilon$.
\end{lemma}
\begin{proof}
It suffices to show that the condition is satisfied with probability at most $\epsilon$ for each $\lambda$.
Let $\hat Z_{\lambda, 1} : = c_\lambda(\hat D) - c_\lambda(T_r \cap \hat D)$ and $\hat Z_{\lambda, 2} = c_\lambda(T_r \cap \hat D)$.
Clearly, $\hat Z_{\lambda,2} \le c_\lambda^g$.
Let $\hat X_p$ be the random variable defined as follows:
let $\hat X_p = 1$ if $p \in \hat D$ and $\hat X_p = 0$ otherwise
%for all maximal element in $T$ and $\bar{x}$.
for element $p$.
All $\hat X_p$ are independent.
%Let $\smallset(\lambda)$ be the maximal element of $T$ and $\bar{x}$ that is small in dimension $\lambda$.
Let $\smallset(\lambda) = P \setminus T_\lambda$.

We calculate the upper bound of the expectation and variance of $\hat Z_{\lambda,1}$ in order to use the Chebyshev-Cantelli bound:
\begin{align}
    \Prob(\hat Z_{\lambda, 1} - \E[\hat Z_{\lambda,1}] \ge t) \le \frac{\V[\hat Z_{\lambda,1}]}{\V[\hat Z_{\lambda,1}] + t^2}.
\end{align}
Notice that
\begin{align}
    \label{eq:lem-e2-decomposition}
    \hat Z_{\lambda,1} = \sum_{p \in \smallset(\lambda)} c_\lambda(p) \hat X_p = \hat Z^{\setminus T}_{\lambda, 1} + \hat Z^T_{\lambda, 1},
\end{align}
where
\begin{align}
    \hat Z_{\lambda,1}^{\setminus T} &= \sum_{p \in \smallset(\lambda) \setminus T} c_{\lambda}(p) \hat X_p,\\
    \hat Z_{\lambda,1}^{T} &= \sum_{p \in T \setminus T_\lambda} c_{\lambda}(p) \hat  X_p.\\
\end{align}
By the order-consistency, $c_{\lambda}$ is a DR-submodular function on $P'$.
Let $C_\lambda^{P'}$ be the multilinear extension of $c_\lambda$ over $K(I(P'))$.
By Theorem~\ref{thm:multilinear-concave},
\begin{align}
    \E[\hat Z_{\lambda,1}^{\setminus T}] = C_\lambda^{P'}(\bar{x}^{1-\epsilon}) &\le (1 - \epsilon) (b_\lambda - %\underline{c}_\lambda
    c_\lambda(T)
    ),
\end{align}
where we used the fact that $b_\lambda - c_\lambda(T)$ is the knapsack capacity of the cost residual problem and $\bar{x}$ is a feasible solution of it.
%Let $\underline{c}_\lambda = c_\lambda(T \setminus T_\lambda) = c_\lambda(T) - c_\lambda(T_\lambda)$.
Since $c_\lambda'(T) := c_\lambda(T \setminus T_\lambda)$ is also a DR-supermodular function on $I(T)$, Therorem~\ref{thm:multilinear-concave} implies that
\begin{align}
    \E[\hat Z_{\lambda,1}^{T}] = C_\lambda'(x_T^{1-\epsilon}) &\le (1 - \epsilon) C_\lambda'(x_T) = (1 - \epsilon )c_\lambda(T \setminus T_\lambda)
    %\underline{c}_\lambda,
\end{align}
where $C_\lambda'$ is the multilinear extension of $c_\lambda'$ over $K(\lat)$.
Hence, 
\begin{align}
\E[\hat Z_{\lambda,1}] \le (1 - \epsilon)(b_\lambda - c_\lambda(T) + c_\lambda(T \setminus T_\lambda)) = (1 -\epsilon)\bar{b}_\lambda.
\end{align}

By (\ref{eq:lem-e2-decomposition}) and $c_{\lambda}(p) < \epsilon^4 b_\lambda$ for $p \in \smallset(\lambda)$, we have
\begin{align} 
    \V[\hat Z_{\lambda, 1}] = \sum_{p \in \smallset(\lambda)} \E[c_{\lambda}(p)^2 \hat X_p^2]
     \le  \sum_{p \in \smallset(\lambda)} \E[c_{\lambda}(p) \hat X_p] \cdot c_{\lambda}(p) \le \epsilon^4 b_\lambda \bar{b}_\lambda.
\end{align}
Here, we used the independence of $\hat X_p$, the bound $|\hat X_p| \le 1$, and
\begin{align}
\sum_{p \in \smallset(\lambda)} \E[c_{\lambda}(p) \hat X_p] \le C_\lambda^{P'}(\bar{x}) + c_\lambda(T \setminus T_\lambda) \le \bar{b}_\lambda.
\end{align}
%The multilinear extension $C_\lambda$ in the above inequality is defined over $K(I(P'))$, not over $K(I(P))$.
%The rest of the proof of \cite[Lemma 2.2]{kulik2009maximizing} holds without any modifications.

We show the statement by using the Chebyshev-Cantelli bound.
If $\bar{b}_\lambda > \epsilon b_\lambda$, then the event $c_\lambda(\hat D) \ge b_\lambda$ is included by $\hat Z_{\lambda,1} - \E[\hat Z_{\lambda,1}] \ge \epsilon \bar{b}_\lambda$.
Therefore, the Chebyshev-Cantelli bound implies that
\begin{align}
    \Prob[c_\lambda(\hat D) \ge b_\lambda] \le \Prob[(\hat Z_{\lambda,1} - \E[\hat Z_{\lambda,1}] \ge \epsilon \bar{b}_\lambda)] &\le \frac{\V[\hat Z_{\lambda,1}]}{\V[\hat Z_{\lambda,1}] + \epsilon^2 \bar{b}_\lambda^2}
    \le \frac{\epsilon^4 b_\lambda \bar{b}_\lambda}{\epsilon^2 \bar{b}_\lambda^2} \le \epsilon.
\end{align}
In the last inequality, we used the assumption $\bar{b}_\lambda > \epsilon b_\lambda$.
Otherwise, $\bar{b}_\lambda \le \epsilon b_\lambda$.
By a similar argument, we have
\begin{align}
    \Prob[c_\lambda(\hat D \setminus T_\lambda) > \epsilon b_\lambda + \bar{b}_\lambda] \le \Prob[\hat Z_{\lambda,1} - \E[\hat Z_{\lambda,1}] > \epsilon b_\lambda] \le \frac{\epsilon^4 b_\lambda \bar{b}_\lambda}{\epsilon^2 b_\lambda^2} \le \epsilon^3 \le \epsilon.
\end{align}
In any cases, the probability that the condition of the step 5  holds is at most  $\epsilon$ for each $\lambda$.
\end{proof}

Let
\begin{align}
    &\hat R_\lambda = \frac{c(\hat D^{(3)})}{b_\lambda}\\
    &\hat R = \max_\lambda \hat R_\lambda.
\end{align}

\begin{lemma}
\label{lem:kulik2.3}
For any $l > 1$,
\begin{align}
    \Prob[\hat R > l] < \frac{|\Lambda|\epsilon^4}{(l-1)^2}.
\end{align}
\end{lemma}
\begin{proof}
The same proof as \cite[Lemma 2.3]{kulik2009maximizing} holds. (With the evaluation of the expectation and the variance of $\hat Z_{\lambda, 1}$ in the previous lemma).
%For the completeness of the paper, we write the proof.
For the sake of completeness, we give a proof here.

We use the same notation as Lemma \ref{lem:kulik2.2}.
By Chebyshev-Cantelli inequality, we have
\begin{align}
    \Prob[\hat R_\lambda > l] &\le \Prob[\hat Z_{\lambda, 1} > l b_\lambda - c_\lambda^g]\\
    &\le \Prob[\hat Z_{\lambda, 1} - \E[\hat Z_{\lambda, 1}] \ge (l-1) b_\lambda])\\
    &\le \frac{\epsilon^4 b_\lambda \bar{b}_\lambda}{(l-1)^2 b_\lambda^2} \le \frac{\epsilon^4}{(l-1)^2}.
\end{align}
By the union bound, we have the statement.
\end{proof}

\begin{lemma}
\label{lem:kulik2.4}
For any integer $l \ge 2$, if $\hat R \le l$, then
\begin{align}
    f(\hat D^{(3)}) \le 2|\Lambda|l f(\OPT).
\end{align}
\end{lemma}
\begin{proof}
In this proof, we omit the superscript $(3)$ to simplify the notation.
We can partition $\hat D$ into $2l|\Lambda|$ subsets $S_1, S_2, \dots, S_{2l|\Lambda|}$ such that $c_\lambda(S_i) \le b_\lambda$ for all $i$ and $\lambda$ by the following reason.
Take some arbitrary partition $\mathcal{S}$ of $\hat D$.
If $|\mathcal{S}| > 2 l |\Lambda|$, then there exists at least two $S_1,S_2 \in \mathcal{S}$ with $c_\lambda(S_i) \le b_\lambda / 2$ for all $i = 1,2$ and $\lambda \in \Lambda$.
Since $c_\lambda(S_1 \cup S_2) \le b_\lambda$ for all $\lambda \in \Lambda$, we can reduce $|\mathcal{S}|$ by taking new partition $\mathcal{S}' = \mathcal{S} \cup \{S_1 \cup S_2\} \setminus \{S_1, S_2\}$.
By continuing this procedure, we have a desired partition of $\hat D$.

%Let $\mathcal{S}_i = \bigcup_{j \le i} S_j$.
We convert $S_i$ into a feasible ideal $S_i'$.
We take a topological order $\hat D = \{p_1, p_2, \dots, p_\alpha\}$ (a larger element appears later).
Hence, $\{p_1, p_2, \dots, p_l\}$ is an ideal for $l \le \alpha$.
Let  $X_{k} = \{p_1, p_2, \dots, p_{k-1}\}$ and $S_i = \{p_{l_1^i}, p_{l_2^i}, \dots, p_{l^i_{|S_i|}}\}$ where $l_1^i < l^i_2 < \dots < l^i_{|S_i|}$.
Then,
\begin{align}
    f(\hat D) \le \sum_{i=1}^{2l|\Lambda|} \sum_{t = 1}^{|S_i|} f_{X_{l_t^i}}(\{p_{l_t^i}\}).
\end{align}
We convert $S_i$ into an ideal $S'_i$ by the following procedure (cf. Algorithm~\ref{alg:pushdown}).
\begin{enumerate}
    \item $S'_{i,(0)} \gets \emptyset$.
    \item Iterate the following procedure for $t =  1, 2,\dots, |S_i|$.
    \begin{enumerate}
        \item Take arbitrary $p'_{l_t^i} \le p_{l_t^i}$ with $s'_t \in \adm{S'_{i,(t-1)}}$.
        \item $S'_{i,(t)} \gets S'_{i,(t-1)} \cup \{p'_{l_t^i}\}$
    \end{enumerate}
    \item $S_i' \gets S'_{i,(|S_i|)}$.
\end{enumerate}

We will see the properties of $S_i'$.
%Let $s_m$ be the $s$ in the $m$-th loop of the above procedure, let $s'_m$ be $s'$, and let $S'_{i, (m)}$ be $S'_i$.
By the construction, $S'_{i, (t)}$ is an ideal for all $t$.
Furthermore, the DR-submodularity of $f$ implies that
\begin{align}
    f_{S'_{i, (t-1)}}(\{p'_{l_t^i}\}) \ge f_{X_{l_t^i}}(\{p_{l_t^i}\}),
\end{align}
for all $t$ since $S'_{i, (t-1)} \subseteq X_{l_t^i}$ by the following reason:
for all $t' < t$, we know that (1) $p_{l^i_{t-1}} \in X_{l^i_t}$, (2) $p'_{l_{t-1}^i} \le p_{l_{t-1}^i}$ and (3) $p'_{l_{t-1}^i} \in X_{l_t^i}$ since $X_{l_t^i}$ is an ideal.
By summing up the above inequality, we have
\begin{align}
    \le f(S'_i) \ge    \sum_{t=1}^{|S_i|} f_{X_{l^i_t}}(\{p_{l^i_t}\}) 
\end{align}
for each $i$.
The order-consistency of $c_\lambda$ implies that $c_{\lambda}(S'_i) \le c_\lambda(S_i) \le b_\lambda$ for all $\lambda$ since $c_\lambda(p_{l_t^i}') < c_\lambda(p_{l_t^i})$.
Hence, $S'_i$ is feasible.
By the optimality of $\OPT$, we have $f(S'_i) \le f(\OPT)$.
These inequalities imply that
\begin{align}
    f(\hat D) \le \sum_{i=1}^{2l|\Lambda|} \sum_{t = 1}^{|S_i|} f_{X_{l^t_i}}(\{p_{l^t_i}\}) \le 2 l|\Lambda|f(\OPT).
\end{align}
\end{proof}

\begin{lemma}
\label{lem:bound-step-5}
\begin{align}
    \E[f(\hat D^{(5)})] \ge (1 - \Theta(\epsilon))(1 - e^{-1})f(\OPT)
\end{align}
\end{lemma}
\begin{proof}
We can prove this lemma by the same argument as \cite[Lemma 2.5]{kulik2009maximizing} with the help of the generalized lemmas (Lemma \ref{lem:kulik2.1}, Lemma \ref{lem:kulik2.2}, Lemma \ref{lem:kulik2.3}, and Lemma \ref{lem:kulik2.4}).
For the sake of completeness, we give a proof here.

Let $B$ be the event where the condition at step 5 is satisfied.
Let $\bar{B}$ be the complementary event of $B$.
By Lemma \ref{lem:kulik2.2}, Lemma \ref{lem:kulik2.3}, and Lemma \ref{lem:kulik2.4}, we have
\begin{align}
    \E[f(\hat D^{(3})] &= E[f(\hat D^{(3)}) \mid B] \Prob[B] + \E[f(\hat D^{(3)}) \mid \bar{B} \land (\hat R < 2)] \Prob[\bar{B} \land (\hat R < 2)]\\
    &+ \sum_{l=1}^\infty \E[f(\hat D^{(3)}) \mid \bar{B} \land (2^l < \hat R \le 2^{l+1})] \Prob[ \bar{B} \land (2^l < \hat R \le 2^{l+1})]\\
    & \le  E[f(\hat D^{(3)}) \mid B] \Prob[B] + \E[f(\hat D^{(3)}) \mid \bar{B} \land (\hat R < 2)] \Prob[\bar{B}]\\
    &+ \sum_{l=1}^\infty \E[f(\hat D^{(3)}) \mid \bar{B} \land (2^l < \hat R \le 2^{l+1})] \Prob[2^l < \hat R]\\
    &\le  E[f(\hat D^{(3)}) \mid B] \Prob[B] + 4 |\Lambda|^2 \epsilon  f(\OPT) + |\Lambda|^2 \epsilon^4 f(\OPT) \sum_{l=1}^\infty \frac{2^{l+2}}{2^{2l-2}}.
\end{align}
In the last inequality, we used the fact that $2^{l} - 1 \ge 2^{l - 1}$ for $l \ge 1$.
Since the last summation is a constant, Lemma \ref{lem:kulik2.1} and the above inequality imply that
\begin{align}
    \E[f(\hat D^{(3)}) \mid B] \Prob[B] \ge (1 - \Theta(\epsilon))(1 - e^{-1})f(\OPT).
\end{align}
Since $\hat D^{(3)} = \hat D^{(5)}$ on the event $B$ and otherwise $\hat D^{(5)} = \emptyset$, we have
\begin{align}
    \E[f(\hat D^{(5)})] = \E[f(\hat D^{(3)}) \mid B] \Prob[B] \ge (1 - \Theta(\epsilon))(1 - e^{-1}) f(\OPT).
\end{align}
\end{proof}

\begin{proof}[Proof of Theorem \ref{thm:rounding_knapsack}]
We first show that $\hat D^{(7)}$ is feasible.
The set $\hat D^{(7)}$ is an ideal because of the property of Algorithm \ref{alg:pushdown}: we can show that $\hat D$ in Algorithm \ref{alg:pushdown} is an ideal at every $t$ by induction.
We show that $c_\lambda(\hat D^{(7)}) \le b_\lambda$ for each $\lambda$.
If $\bar{b}_\lambda > \epsilon b_\lambda$, then $c_\lambda(\hat D_1 \cup \hat D_2) \le b_\lambda$ after step 5 and nothing happens at step 6 and step 7.
Hence, $c_\lambda(\hat D^{(7)}) \le b_\lambda$.
Else, $\bar{b}_\lambda \le \epsilon b_\lambda$.
In this case, we can make $c_\lambda(\hat D_2^{(6)}) = 0$ at step 6 if necessary.
At step 7, Algorithm \ref{alg:pushdown} does not increase $c_\lambda$ since $c_\lambda(p_t') \le c_\lambda(p_t)$ at each $t$ by the order-consistency of $c_{\lambda}$.
In conclusion,
\begin{align}
    c_\lambda(\hat D^{(7)}) \le c_\lambda(\hat D_1^{(5)}) \le \bar{b}_\lambda + \epsilon b_\lambda \le 2 \epsilon b_\lambda \le b_\lambda,
\end{align}
for sufficiently small $\epsilon$ ($< 1/2$).
In conclusion, we have shown that $\hat D^{(7)}$ is feasible.

We next show that $f(\hat D^{(7)}) \ge f(\hat D^{(5)}) - |\Lambda| \epsilon^3 M$.
Together with Lemma~\ref{lem:bound-step-5}, it is sufficient for the desired approximation guarantee.
We first prove that
\begin{align}
    \label{eq:bound-Dto6sub2}
    f(\hat D^{(6)}_2) \ge f(\hat D^{(5)}_2) - |\Lambda| \epsilon^3 M.
\end{align}
If $\bar{b}_\lambda > \epsilon b_\lambda$, then nothing happens at step 6 and the above inequality holds.
Else, $\bar{b}_\lambda \le \epsilon b_\lambda$.
The excess of the knapsack $c_\lambda$ is at most $\epsilon b_\lambda$ due to the definition of step 5.
Thus, the removal procedure at step 6 occurs at most $\epsilon^3$ times for each $\lambda$ since
%Thus, at most $\epsilon^3$  at step 6 for each $\lambda$ because
$c_{\lambda}(p) > \epsilon^4 b_\lambda$ for big element $p$.
%Therefore, the number of the removed elements is at most $|\Lambda|\epsilon^3$.
%by summing up $\epsilon^3$ for all $\lambda$.
Therefore, the removal procedure at step 6 occurs at most $|\Lambda|\epsilon^3$ times in total.
At each removal, we can prove that $f(\hat D_2)$ decreases by $M$ at most.
Indeed, the difference of $\hat D_2$ is at most $I_p$ when big element $p$ is removed.
Together with DR-submodularity of $f$, we can show that $f(\hat D_2)$ decreases at most $f(I_p)$, which is not greater than $M = \max_p f(I_p)$.

By using this bound for $f(\hat D_2^{(6)})$, we prove $f(\hat D^{(7)}) \ge f(\hat D^{5}) - |\Lambda| \epsilon^3 M$.
Let $\hat D_1^{(7)} = \hat D^{(7)} \setminus \hat D_2^{(6)}$.
%\COMM{SN}{$\hat D_2$のほうはどのstepでも自然に定義されている．step 7のはstep6で減らした後のもの．$\hat D_1^{(7)}$は$\hat D_1^{(6)}$を押し下げ（push down）てidealにしたもの．}
Notice that
\begin{align}
    f(\hat D^{(5)}) = f(\hat D^{(5)}_2) + f_{\hat D^{(5)}_2}(\hat D_1^{(5)}),\\
    f(\hat D^{(7)}) = f(\hat D^{(6)}_2) + f_{\hat D^{(6)}_2}(\hat D_1^{(7)}).
\end{align}
By the DR-submodularity of $f$ and the definition of Algorithm \ref{alg:pushdown}, we can prove that
$f_{\hat D^{(5)}_2}(\hat D_1^{(5)}) \le f_{\hat D^{(6)}_2}(\hat D_1^{(7)})$.
Indeed, at each iteration of Algorithm \ref{alg:pushdown}, the DR-submodularity implies that
$f_{\hat D}(p'_{t}) \ge f_{\hat D_2'}(p_{t})$,
where $\hat D$ is considered after the $t$-th iteration of Algorithm \ref{alg:pushdown} and $\hat D_2' = \hat D_2^{(5)} \cup \{p_{t'} \mid t' < t\}$.
By summing this inequality, we have 
$f_{\hat D^{(5)}_2}(\hat D_1^{(5)}) \le f_{\hat D^{(6)}_2}(\hat D_1^{(7)})$.
In conclusion, we have
\begin{align}
    f(\hat D^{(7)}) &= f(\hat D^{(6)}_2) + f_{\hat D^{(6)}_2}(\hat D_1^{(7)})\\
    &\ge f(\hat D^{(5)}_2) - |\Lambda|\epsilon^3M + f_{\hat D^{(5)}_2}(\hat D_1^{(5)})\\
    &= f(\hat D^{(5)}) - |\Lambda| \epsilon^3M.
\end{align}
\end{proof}

\bibliographystyle{plain}
\bibliography{main}

%\appendix
%\input{appendix.tex}

\end{document}